\documentclass[11pt]{article}

\usepackage[margin=1in]{geometry}
\usepackage{amsfonts,amssymb,amsmath,amsthm,mathtools}
\usepackage[usenames,dvipsnames]{xcolor}
\usepackage{booktabs}
\usepackage{enumitem}
\usepackage{tocloft}
\usepackage{thm-restate}
\usepackage{dsfont}
\usepackage{tikz}
\usepackage{bbm}

\usepackage[pagebackref]{hyperref}
\hypersetup{
    pdftitle={}, 
    pdfauthor={}, 
    colorlinks=true, 
    linkcolor=blue, 
    citecolor=blue, 
    filecolor=blue, 
    urlcolor=blue 
}

\allowdisplaybreaks

\renewcommand{\backref}[1]{}

\renewcommand{\backrefalt}[4]{%
\small
\ifcase #1 %
\or
[p.\ #2]%
\else
[pp.\ #2]%
\fi}

\makeatletter
\newcommand{\para}{%
  \@startsection{paragraph}{4}%
  {\z@}{1.5ex \@plus .5ex \@minus .3ex}{-1em}%
  {\normalfont\normalsize\bfseries}%
}
\makeatother

\newtheorem{theorem}{Theorem}
\newtheorem{lemma}[theorem]{Lemma}

\newtheorem{fact}[theorem]{Fact}
\newtheorem{claim}[theorem]{Claim}

\theoremstyle{definition}
\newtheorem{definition}[theorem]{Definition}

\newtheoremstyle{part}
  {-0.2\topsep}   
  {\topsep}   
  {\itshape}  
  {0pt}       
  {\bfseries} 
  {.}         
  {5pt plus 1pt minus 1pt} 
  {}          

\theoremstyle{part}
\newtheorem{factpart}{Fact}[theorem]

\newcommand{\eq}[1]{\hyperref[eq:#1]{(\ref*{eq:#1})}}
\renewcommand{\sec}[1]{\hyperref[sec:#1]{Section~\ref*{sec:#1}}}
\newcommand{\thm}[1]{\hyperref[thm:#1]{Theorem~\ref*{thm:#1}}}
\newcommand{\lem}[1]{\hyperref[lem:#1]{Lemma~\ref*{lem:#1}}}
\newcommand{\prop}[1]{\hyperref[prop:#1]{Proposition~\ref*{prop:#1}}}
\newcommand{\cor}[1]{\hyperref[cor:#1]{Corollary~\ref*{cor:#1}}}
\newcommand{\fig}[1]{\hyperref[fig:#1]{Figure~\ref*{fig:#1}}}
\newcommand{\tab}[1]{\hyperref[tab:#1]{Table~\ref*{tab:#1}}}
\newcommand{\alg}[1]{\hyperref[alg:#1]{Algorithm~\ref*{alg:#1}}}
\newcommand{\app}[1]{\hyperref[app:#1]{Appendix~\ref*{app:#1}}}
\newcommand{\defn}[1]{\hyperref[def:#1]{Definition~\ref*{def:#1}}}
\newcommand{\clm}[1]{\hyperref[clm:#1]{Claim~\ref*{clm:#1}}}
\newcommand{\fct}[1]{\hyperref[fact:#1]{Fact~\ref*{fact:#1}}}

\newcommand*{\fullref}[1]{\hyperref[{#1}]{\autoref*{#1}:~\nameref*{#1}}}

\newcommand*{\fullbref}[1]{\hyperref[{#1}]{\autoref*{#1} (\nameref*{#1})}}

\newcommand{\B}{\{0,1\}}

\newcommand{\AND}{\textsc{And}}

\newcommand{\IP}{\textsc{IP}}

\newcommand{\XOR}{\textsc{Xor}}

\newcommand{\X}{\mathcal{X}}
\newcommand{\Y}{\mathcal{Y}}

\newcommand{\U}{\mathcal{U}}

\newcommand{\CS}{\mathrm{CS}}

\newcommand{\tO}{\widetilde{O}}
\newcommand{\tOmega}{\widetilde{\Omega}}

\renewcommand{\(}{\left(}
\renewcommand{\)}{\right)}
\newcommand{\<}{\langle}
\renewcommand{\>}{\rangle}
\newcommand{\id}{\mathbbm{1}}

\DeclareMathOperator{\rank}{rk}
\DeclareMathOperator{\arank}{\rank_{1/3}}
\newcommand{\alogrank}{\log\arank}
\DeclareMathOperator{\Size}{size}
\DeclareMathOperator{\disc}{disc}

\DeclareMathOperator{\dom}{{dom}}
\DeclareMathOperator{\polylog}{polylog}

\DeclareMathOperator{\Dc}{D}
\DeclareMathOperator{\R}{R}
\newcommand{\Q}{\mathrm{Q}^*}
\newcommand{\Qn}{\mathrm{Q}}

\newcommand{\G}{\mathcal{G}}
\newcommand{\x}{\mathbf{x}}
\newcommand{\y}{\mathbf{y}}
\newcommand{\fu}{\mathbf{u}}
\newcommand{\fv}{\mathbf{v}}

\newcommand{\defeq}{\coloneqq} 

\newcommand{\eps}{\varepsilon}
\renewcommand{\epsilon}{\varepsilon}
\renewcommand{\Pr}{\mathrm{Pr}}
\newcommand{\E}{\mathcal{E}}
\newcommand{\F}{\mathrm{F}}
\newcommand{\I}{\mathbb{I}}

\newcommand{\VR}{\Delta}

\newcommand{\err}{\mathrm{err}}

\newcommand {\fn} [2] {\ensuremath{ #1 \minusspace \br{ #2 } }}

\newcommand {\minusspace} {\: \! \!}

\newcommand {\br} [1] {\ensuremath{ \left( #1 \right) }}

\newcommand {\bra} [1] {\ensuremath{ \left\langle #1 \right| }}
\newcommand {\ket} [1] {\ensuremath{ \left| #1 \right\rangle }}
\newcommand {\ketbratwo} [2] {\ensuremath{ \left| #1 \middle\rangle \middle\langle #2 \right| }}
\newcommand {\ketbra} [1] {\ketbratwo{#1}{#1}}
\newcommand {\Tr} {\ensuremath{ \mathrm{Tr} }}
\newcommand {\norm} [1] {\ensuremath{ \left\| #1 \right\| }}
\newcommand {\normsub} [2] {\ensuremath{ \norm{#1}_{#2} }}
\newcommand {\onenorm} [1] {\normsub{#1}{1}}
\newcommand {\cspace} [1] {\ensuremath{\mathnormal{#1}}}
\newcommand {\set} [1] {\ensuremath{ \left\lbrace #1 \right\rbrace }}
\newcommand {\relent} [2] {\fn{\mathrm{S}}{#1 \middle\| #2}}
\newcommand {\ent} [1] {\fn{\mathrm{S}}{#1}}

\newcommand{\expec}{\mathbb{E}}

\def\D{\mathcal{D}}
\def\L{\mathcal{L}}

\def\H{\mathcal{H}}
\newcommand{\BR}{\mathrm{B}}

\newcommand{\QCC}{\mathrm{QCC}}

\newcommand{\tX}{\widetilde{X}}
\newcommand{\tU}{\widetilde{U}}
\newcommand{\tY}{\widetilde{Y}}
\newcommand{\tV}{\widetilde{V}}

\begin{document}

\title{Separating quantum communication and approximate rank}

\author{
Anurag Anshu\footnote{Centre for Quantum Technologies, National University of Singapore, Singapore. \texttt{a0109169@u.nus.edu}} \qquad
Shalev Ben-David\footnote{Massachusetts Institute of Technology. \texttt{shalev@mit.edu}} \qquad
Ankit Garg\footnote{Microsoft Research New England. \texttt{garga@microsoft.com}} \\[1.5ex]
Rahul Jain\footnote{Centre for Quantum Technologies, National University of Singapore and MajuLab, UMI 3654, 
Singapore. \texttt{rahul@comp.nus.edu.sg}} \qquad 
Robin Kothari\footnote{Center for Theoretical Physics, Massachusetts Institute of Technology. \texttt{rkothari@mit.edu}} \qquad 
Troy Lee\footnote{SPMS, Nanyang Technological University and Centre for Quantum Technologies and MajuLab, UMI 3654, Singapore. {\tt troyjlee@gmail.com}}
}

\hypersetup{pageanchor=false} 
\date{}
\maketitle

\begin{abstract}
One of the best lower bound methods for the quantum communication complexity of a function $H$ (with or without shared 
entanglement) is the logarithm of the approximate rank of the communication matrix of $H$. 
This measure is essentially equivalent to the approximate $\gamma_2$ norm and generalized discrepancy, and subsumes several other 
lower bounds. 
All known lower bounds on quantum communication complexity in the general unbounded-round model can be shown
via the logarithm of approximate rank, and it was an open problem to give any separation at all between quantum 
communication complexity and the logarithm of the approximate rank.

In this work we provide the first such separation: We exhibit a total function $H$ with quantum communication complexity almost quadratically larger than the logarithm of its approximate rank. We construct $H$ using the communication lookup 
function framework of Anshu et al.~(FOCS 2016) based on the cheat sheet framework of Aaronson et al.~(STOC 2016).  
From a starting function $F$, this framework defines a new function $H=F_\G$.
Our main technical result is a lower bound on the quantum communication complexity of $F_\G$ in terms of the 
discrepancy of $F$, which we do via quantum information theoretic arguments.  
We show the upper bound on the approximate rank of $F_\G$
by relating it to the Boolean circuit size of the starting function $F$.
\end{abstract}


\thispagestyle{empty}
\clearpage
\section{Introduction}
\label{sec:intro}
\hypersetup{pageanchor=true} 
\setcounter{page}{1}
Communication complexity studies how much two parties Alice and Bob need to communicate in order to compute a 
function when each party only has partial knowledge of the input.  The model of quantum communication complexity 
allows the players to send quantum messages back and forth, and measures the total number of qubits that need to 
be exchanged in order to compute the function.  
Communication complexity has become a fundamental 
area in theoretical computer science with applications to circuit complexity, data structures, streaming algorithms, 
property testing, and linear and semi-definite programs. 
Many of these applications require showing communication complexity \emph{lower bounds}, which raises the importance of studying lower bound techniques in communication complexity. 

In this paper we study lower bounds on quantum communication complexity. For a two-party function $F: \X \times \Y \rightarrow \{0,1\}$, we denote by $\Qn(F)$ the minimum number of qubits needed by a quantum protocol to compute $F$ with error probability at most $1/3$.  

One of the strongest lower bounds on $\Qn(F)$ comes by viewing $F$ as a Boolean $|\X|\times|\Y|$ matrix, known as the communication matrix, 
which we will also denote by $F$.  The approximate rank of $F$, denoted $\arank(F)$, is 
the minimum rank of a matrix $\tilde F$ that is entrywise close to $F$, that is, satisfying $\ell_\infty(\tilde F - F) \le 1/3$.  Building on the work of Kremer \cite{Kre95} and Yao \cite{Yao93}, Buhrman and de Wolf \cite{BdW01} showed that 
$\Qn(F)=\Omega(\alogrank(F))$.  Later, it was shown that approximate rank can also be used to lower bound 
quantum communication complexity with shared entanglement, denoted $\Q(F)$.  More precisely, 
$\Q(F)=\Omega(\alogrank(F))- O(\log \log (|\X| \cdot |\Y|))$ \cite{LS08c}.  As this paper studies quantum communication complexity lower bounds, we will focus on the measure $\Q(F)$, which makes our results stronger.

The logarithm of the approximate rank dominates nearly all other lower bounds on quantum 
communication complexity, including the discrepancy method \cite{Kre95}, the approximate trace norm~\cite{Raz03,LS07}, the generalized discrepancy method 
\cite{Kla07, Raz03, She11}, and the approximate
$\gamma_2$ norm bound \cite{LS07}.\footnote{In fact, the generalized discrepancy method, 
logarithm of approximate $\gamma_2$ norm, and logarithm of approximate rank are all equivalent, up 
to constant mutliplicative factors and an additive logarithmic term.} In fact, to the best of our knowledge, all known lower bounds for general two-way quantum communication can be obtained using approximate rank. Besides being a powerful lower bound method, approximate rank is a robust measure posessing several desirable properties such as error reduction, direct sum and strong direct product theorems~\cite{She12}, and an optimal lifting theorem~\cite{She11,SZ09b}.

Given our current state of knowledge, it is consistent that $\Q(F) = O(\alogrank(F))$ 
for every function $F$, that is, the logarithm of the approximate rank \emph{characterizes} quantum communication complexity. 
As it is widely believed that this is not the case, this state of affairs points to the 
limitations of our current lower bound techniques for quantum communication complexity.  

In this paper, we show the first superlinear separation between quantum communication complexity and the logarithm of 
the approximate rank.
\begin{restatable}{theorem}{separation}
\label{thm:separation}
	There is a family of total functions $F:\X\times\Y\to\B$ with
	$\Q(F)=\tOmega\Bigl(\log^2 \arank(F)\Bigr)$.
\end{restatable}

As far as we are aware, \autoref{thm:separation} is the first superlinear separation between quantum communication complexity and the logarithm of the approximate rank even for \emph{partial} functions, which are functions defined only on a subset of the domain $\X \times \Y$.\footnote{For partial functions, we require the approximate low-rank decomposition of the communication matrix to take values between 0 and 1 even on inputs on which the function is undefined. Without this constraint it is easy to construct large partial function separations.}

One alternative to approximate rank for showing lower bounds on quantum communication complexity is the 
recently introduced quantum information complexity \cite{Tou15}.  This bound has been shown to dominate the logarithm of 
the approximate rank \cite{Braverman15}, and has nice properties like characterizing amortized quantum communication 
complexity.  The quantum information complexity, however, is difficult to bound for an explicit function and has not yet been 
used to show a new lower bound in the general unbounded-round model of quantum communication complexity.  

By analogy with the log rank conjecture, which postulates that $\Dc(F) = O(\polylog(\rank(F)))$, where $\Dc(F)$ is the deterministic communication complexity of $F$, it is natural to state 
an approximate log rank conjecture.  The quantum version of the approximate log rank conjecture states 
$\Q(F) = O(\polylog(\arank(F)))$.  Our results show that the exponent of the logarithm in such a statement must 
be at least 2.  The largest gap we currently know between $\Dc(F)$ and 
$\log \rank(F)$ is also quadratic \cite{GPW15}. One could also consider a randomized version of the log rank 
conjecture, stating $\R(F) = O(\polylog(\arank(F)))$, where $\R(F)$ is the $1/3$-bounded-error randomized 
communication complexity.  This conjecture is actually known to imply the usual deterministic log rank conjecture 
\cite{KMS+14}. The largest known gap between $\R(F)$ and $\log \arank(F)$ is $4\text{th}$ power \cite{GJPW15}. 

Our separation is established using quantum information theoretic arguments to lower bound quantum communication complexity of a particular family of functions known as lookup functions, introduced in \cite{ABB+16a}. We use Boolean circuit size to upper bound the logarithm of approximate rank of lookup functions.
We now provide an overview of lookup functions and our proof techniques.

\subsection{Techniques}
Many questions in communication complexity have analogs in the (usually simpler) model of query complexity.    The query 
complexity quantity that is analogous to approximate rank is the approximate 
polynomial degree.  Using the quantum adversary lower bound, Ambainis \cite{Amb03} gave a function $f$ with an $n$ versus 
$n^{1.32}$ separation between its approximate polynomial degree and quantum query complexity.  This result is the main 
reason for the belief that there should also be a separation between the logarithm of approximate rank and quantum 
communication complexity.  One way to do this would be to ``lift'' the quantum query lower bound for $f$ into a 
quantum communication lower bound for a related communication problem by composing $f$ with an appropriate communication gadget.  While such a lifting theorem is known for the approximate polynomial degree \cite{She11,SZ09b}, 
it remains an open question to show a lifting theorem for quantum query complexity or the quantum adversary method.  The lack of an analog of the  
adversary lower bound in the setting of quantum communication complexity is part of the difficulty of separating the 
logarithm of approximate rank and quantum communication complexity.

There has recently been a great deal of progress in showing new separations between complexity measures in query 
complexity \cite{GPW15,ABB+16,ABK16}.
The work in query complexity most closely related to ours is the \emph{cheat sheet} method of Aaronson et al.~\cite{ABK16}.
The cheat sheet method is a way to transform a function $f$ into its ``cheat sheet'' version $f_{\CS}$ so that, for some 
complexity measures, $f_\CS$ retains the hardness of $f$, while other complexity measures are drastically reduced 
by this transformation.  Among other things, Aaronson et al.~\cite{ABK16} use this method to improve 
Ambainis' separation and give a 4th power separation 
between quantum query complexity and approximate polynomial degree.

\cite{ABB+16a} generalize the cheat sheet method to communication complexity.  They are able to lift several query 
results of \cite{ABK16} to communication complexity, such as an example of a total function with a super-quadratic 
separation between its randomized and quantum communication complexities.  They do this by introducing the idea 
of a \emph{lookup} function.  To motivate a lookup function, consider first a communication version of the familiar 
address function.  Alice receives inputs $x \in \{0,1\}^c$ and $u_0, \ldots, u_{2^c-1} \in \{0,1\}$ and Bob receives 
$y \in \{0,1\}^c$ and $v_0, \ldots, v_{2^c-1} \in \{0,1\}$.  The desired output is found by interpreting 
$x \oplus y$ as the binary representation of a number $\ell \in \{0, \ldots, 2^c-1\}$ and outputting $u_\ell \oplus v_\ell$.  

The $(F, \G)$ lookup function $F_\G$ is defined by a function $F: \X \times \Y \rightarrow \{0,1\}$ and a function family   $\mathcal{G}=\{G_0, \ldots, G_{2^c-1}\}$, with $G_i: (\X^c \times \{0,1\}^m) \times (\Y^c \times \{0,1\}^m) 
\rightarrow \{0,1\}$.  Alice receives input $\x=(x_1, \ldots, x_c) \in \X^c$ and $u_0, \ldots, u_{2^c-1} \in \B^m$ and Bob receives 
inputs $\y=(y_1, \ldots, y_c) \in \Y^c$  and $v_0, \ldots, v_{2^c-1}\in \B^m$.  Now the address is determined by interpreting
$(F(x_1, y_1), \ldots, F(x_c, y_c)) \in \{0,1\}^c$ as an integer $\ell \in \{0, \ldots, 2^c-1\}$ and the goal of the 
players is to output $G_\ell((\x,u_\ell), (\y,v_\ell))$.  Note that, in contrast to the case with the address function, in 
a lookup function, $G_{\ell}$ can depend on $\x$ and $\y$.  This is the source of difficulty in showing lower bounds 
for lookup functions, and also key to their interesting properties.

\para{Lower bound.} The main result of \cite{ABB+16a} showed that, given some mild restrictions on the family of functions $\G$, the randomized 
communication complexity of $F_\G$ is at least that of $F$.  Our main result shows that, given mild restrictions 
on the function family $\G$, if there is a quantum protocol with $q$ qubits of communication for $F_\G$, then there is a $q$ 
qubit protocol for $F$ with \emph{non-negligible bias}.  Because of the round-by-round nature of our quantum 
information theoretic argument, the success probability of the quantum protocol for $F$ decays with the number of rounds 
of the quantum protocol for $F_\G$.  Thus to apply this theorem, we need to start with a function $F$ that has 
high quantum communication complexity even for protocols with small bias.  As the discrepancy method lower bounds 
quantum communication complexity even with small bias, we can informally state our main theorem as follows. 
\begin{restatable}[Informal restatement of \cor{lower}]{theorem}{CSinformal}
\label{thm:CSinformal}
For any $(F, G)$ lookup function $F_\G$, provided $\G$ satisfies certain mild technical conditions, $\Q(F_\G) = \Omega(\log (1/\disc(F)))$.
\end{restatable}

Let us call such theorems, where we lower bound the complexity of a lookup function $F_\G$ (or a cheat sheet function $f_\CS$) in terms of a measure of the original function $F$ (or $f$), ``cheat sheet theorems.'' Essentially optimal cheat sheet theorems have been shown in a number of computational models such as deterministic, randomized, and quantum query complexity ~\cite{ABK16} and randomized communication complexity~\cite{ABB+16a}. Cheat sheet theorems are in spirit similar to joint computation results such as direct sum and direct product theorems \cite{Barak2013, BRWY_dp, BW15, Dru12, LR13, She12, Tou15}.\footnote{One point of difference is that in direct sum and direct product theorems, the lower bounds on the amount of resources (query, communication, etc.) usually scale with $c$, the number of copies of the function $F$. In the cheat sheet theorem we prove (and also in prior works), the lower bounds do not scale with $c$. This is due to the fact that the value of $c$ is usually small in our applications.
} 
Direct sum and direct product theorems are widely applicable tools and are often an important goal by themselves. 
Cheat sheet theorems have become useful tools recently and for example, the cheat sheet theorems proven in \cite{ABK16} were later used in \cite{AKK16}. 
We hope that our quantum cheat sheet theorem will find further applications.

We now provide a high-level overview of the proof of our quantum cheat sheet theorem. We would like to rule out the existence of a quantum protocol $\Pi$ that solves the lookup function $F_\G$ and whose communication cost is much smaller than the quantum communication complexity of $F$ (with inverse polynomial bias, for technical reasons explained below). Since $\Pi$ has small communication cost, during the course of the protocol Alice and Bob do not know the value of the index $\ell = (F(x_1, y_1), \ldots, F(x_c, y_c))$. Also since there are too many cells in the array, which has length $2^c \gg \Q(F)$, and $\Pi$ has small communication cost, Alice and Bob cannot talk about too many cells of the array. We first show that these two conditions imply that Alice and Bob have little information about the contents of the correct cell of the other player's array, i.e., Alice has little information about $v_\ell$ and Bob has little information about $u_\ell$. 

In the hypothesis of the theorem, we assume that $G_\ell$ satisfies a \emph{nontriviality condition}: this states that 
$G_\ell(\x,\y,u_\ell,v_\ell)$ takes both values $0$ and $1$ as $(u_\ell,v_\ell)$ range over all possible values.  
Thus the fact that Alice has little information about $v_\ell$ and Bob has little information about $u_\ell$ sounds 
like we have reached a contradiction already.  The issue is that we do not have any control over the 
\emph{bias} of $G_\ell$.
This situation is reminiscent of the quantum information theoretic arguments in the proof of quantum communication complexity lower bounds for the disjointness function \cite{JRS03}. In that case, one has to argue that a quantum protocol that solves the AND function on $2$ bits exchanges non-trivial amount of information even on distributions which are extremely biased towards the AND being $0$. We use similar arguments (namely the {\em quantum cut-and-paste} argument) to obtain a contradiction for our lookup function. Quantum cut-and-paste arguments usually have a round dependence (which is provably needed for the disjointness lower bound) but which may not be needed for our lookup function. Improving our quantum cheat sheet theorem or proving that it is tight remains an excellent open question. 

At a high level our proof follows the same strategy as the proof for randomized communication complexity in \cite{ABB+16a}, but the implementation of the steps of the argument is different due to the quantum nature of the protocol. A quantum communication protocol presents several challenges, such as the fact that there is no notion of a communication transcript, since it is not possible to store all the quantum messages exchanged during the protocol. Hence arguments that applied to the overall communication transcript do not work in the quantum setting. Several technical lemmas, such as the Markov chain property of classical communication protocols used in \cite{ABB+16a}, fail to hold in the quantum setting.

\para{Upper bound.} We devise a general technique for proving upper bounds on the logarithm of approximate rank of lookup functions for carefully constructed function families $\G$. Given a circuit $\mathcal{C}$ for $F$, a cell in the array tries to certify the computation of $F$ by the circuit $\mathcal{C}$. More formally, $G_\ell(\x,\y,u_\ell,v_\ell) = 1$ iff $(F(x_1, y_1), \ldots, F(x_c, y_c))=\ell$ and $u_\ell \oplus v_\ell$ provides the values of the inputs and outputs to all the gates in $\mathcal{C}$ for each of the $c$ different evaluations of $\mathcal{C}$ on inputs $(x_1,y_1),\ldots,(x_c,y_c)$. 
We show 
that a small circuit for $F$ implies a good upper bound on the approximate rank of the lookup function $F_\G$.

\begin{theorem}[Informal restatement of \thm{upper}]\label{thm:upperinformal} For any Boolean function $F$, there exists a family of functions $\G$ satisfying certain nontrivality conditions such that the lookup function 
$F_\G$ satisfies $\log \arank(F_\G) = \tO(\sqrt{\Size(F)})$.
\end{theorem}

Here $\Size(F)$ denotes the size of the smallest circuit (i.e., the one with the least number of gates) for $F$ over some constant-sized gate set, such as the set of all $2$-bit gates. The high level idea for the upper bound is the following. Suppose an all-knowing prover Merlin provided Alice and Bob the value $\ell = (F(x_1, y_1), \ldots, F(x_c, y_c))$. Then they can ``unambiguously"  verify  Merlin's answer with a small amount of quantum communication. Essentially they look at the $\ell^{\text{th}}$ cell of the array and try to find an inconsistency in the circuit values. This can then be done with quadratically less communication by a quantum protocol by using a distributed version of Grover's algorithm~\cite{Gro96,BCW98}. We then show that this sort of upper bound on ``unambiguously certifiable quantum communication" provides an upper bound on the log of approximate rank of the lookup function $F_{\G}$. A similar upper bound was also used in the query complexity separations of~\cite{ABK16}.

Putting these upper and lower bounds together, if we choose $F$ to be the inner product function, which has exponentially small discrepancy and linear circuit size, \thm{CSinformal} and \thm{upperinformal} give us the desired quadratic separation between quantum communication complexty and the log of approximate rank for a lookup function $F_{\G}$.

One intriguing aspect of \thm{upperinformal} is that if one can prove lower bounds on $\log \arank(F_\G) \gg \sqrt{n}$ for {\em every} nontrivial function family $\G$, then one proves nontrivial circuit lower bounds for $F$! This theorem is similar in flavor to the theorem \cite{LLS06, Rei11} that the square of the quantum query complexity of a function $f$ is a lower 
bound on the formula size of $f$. 
It might seem hopeless to prove a lower bound on $\log \arank(F_\G)$ for every nontrivial function family $\G$, but this is exactly what our quantum cheat sheet theorem achieves for quantum communication complexity, and what the results of \cite{ABB+16a} achieve for randomized communication complexity.

\section{Preliminaries and notation}
\label{sec:prelim}

We will use $X,Y,Z$ to denote random variables as well as their distributions. $x \leftarrow X$ will stand for $x$ being sampled from the distribution of $X$. For joint random variables $XY$, $Y^x$ will denote the distribution of $Y|X=x$.

We now state some classical complexity measures that will be used in this paper. We define quantum measures in more detail in \sec{qinf} and \sec{comm}. We first formally define approximate rank.

\begin{definition}[Approximate rank]
Let $\epsilon \in [0,1/2)$ and $F$ be an $|\X| \times |\Y|$ matrix.  The $\epsilon$-approximate rank of $F$ is 
defined as
\[
\rank_\epsilon(F) = \min_{\tilde F}\  \{\rank(\tilde F) : \forall x\in\X, y\in \Y, \, |\tilde F(x,y) - F(x,y)| \le \epsilon\}
\] 
\end{definition}

As discussed in the introduction, approximate rank lower bounds bounded-error quantum communication complexity with shared entanglement. It also lower bounds $\eps$-error quantum communication~\cite{LS08c}:

\begin{fact}
\label{fact:arankQ}
For any two-party function $F:\X \times \Y \to \B$ and $\eps\in[0,1/3]$,  we have $\Q_\eps(F) = \Omega(\log \rank_\eps(F)) - O(\log\log(|\X|\cdot |\Y|))$.
\end{fact}

Another classical lower bound measure that we use is the discrepancy of a function~\cite{KN06}.

\begin{definition}[Discrepancy]
\label{def:disc}
Let $F$ be an $|\X| \times |\Y|$ Boolean-valued matrix and $P$ a probability distribution over $\X \times Y$.  
The discrepancy of $F$ with respect to $P$ is 
\[
\disc_P(F) = \max_R \left | \sum_{(x,y) \in R} P(x,y) (-1)^{F(x,y)} \right | \enspace ,
\]
where the maximum is taken with respect to all combinatorial rectangles $R$.  The discrepancy of $F$, 
denoted $\disc(F)$, is defined as $\disc(F) = \min_P \disc_P(F)$, where the minimum is taken over all 
probability distributions $P$.
\end{definition}

The discrepancy bound lower bounds not only bounded-error quantum communication complexity, but also quantum communication complexity with error exponentially close (in the discrepancy) to $1/2$.
More precisely, we have the following~\cite{Kre95, LS07}.

\begin{theorem}
\label{thm:discQ}
Let $F: \X \times \Y \rightarrow \{0,1\}$ be a two-party function and $\eps\in[0,1/2)$. Then
\[
\Q_\epsilon(F) = \Omega \left(  \log \frac{1-2\epsilon}{\disc(F)} \right).
\]
\end{theorem}

Finally we define the Boolean circuit size of a function. To do this, we first fix a gate set, say the set of all gates with 2 input bits (although we could have chosen any constant instead of 2). 

\begin{definition}[Circuit size]
For a function $F:\B^n \times \B^m \to \B$, we define $\Size(F)$ to be the size (i.e., number of gates) of the smallest circuit over the gates set of all 2-input Boolean gates that computes $F$.
\end{definition}

Note that here the encoding of Alice's and Bob's input is important, since different input representations may yield different sized circuits, unlike in communication complexity. When we use this size measure, we only deal with functions defined on bits where the input encoding is clearly specified.

\subsection{Quantum Information}
\label{sec:qinf}
We now introduce some quantum information theoretic notation. We assume the reader is familiar with standard notation in quantum computing~\cite{NC00,Wat16}.


Let $\H$ be a finite-dimensional complex Euclidean space, i.e.,  $\mathbb{C}^n$ for some positive integer $n$ with the usual complex inner product $\langle \cdot, \cdot \rangle$, which is defined as $\langle u,v \rangle = \sum_{i=1}^n u_i^* v_i$. We will also refer to $\H$ as a Hilbert space. We will usually denote vectors in $\H$ using braket notation, e.g., $|\psi\> \in \H$.

The $\ell_1$ norm (also called the trace norm) of an operator $X$ on $\H$ is $\onenorm{X}\defeq \Tr (\sqrt{X^{\dag}X})$, which is also equal to (vector) $\ell_1$ norm of the vector of singular values of $X$. 

A {\em quantum state} (or a {\em density matrix} or simply a {\em state}) $\rho$ is a positive semidefinite matrix on $\H$ with $\Tr(\rho)=1$. The state $\rho$ is said to be a {\em pure state} if its rank is $1$, or equivalently if $\Tr(\rho^2)=1$, and otherwise it is called a {\em mixed state}. 
 Let $\ket{\psi}$ be a unit vector on $\H$, that is $\langle \psi|\psi \rangle=1$.  With some abuse of notation, we use $\psi$ to represent the vector $|\psi\>$ and also the density matrix $\ketbra{\psi}$, associated with $\ket{\psi}$. Given a quantum state $\rho$ on $\H$, the {\em support of $\rho$}, denoted $\text{supp}(\rho)$ is the subspace of $\H$ spanned by all eigenvectors of $\rho$ with nonzero eigenvalues.

A {\em quantum register} $A$ is associated with some Hilbert space $\H_A$. Define $|A| \defeq \log \dim(\H_A)$. Let $\L(A)$ represent the set of all linear operators on $\H_A$. We denote by $\D(A)$ the set of density matrices on the Hilbert space $\H_A$. We use subscripts (or superscripts according to whichever is convenient) to denote the space to which a state belongs, e.g, $\rho$ with subscript $A$ indicates $\rho_A \in \H_A$. If two registers $A$ and $B$ are associated with the same Hilbert space, we represent this relation by $A\equiv B$.  For two registers $A$ and $B$, we denote the combined register as $AB$, which is associated with Hilbert space $\H_A \otimes \H_B$.  For two quantum states $\rho\in \D(A)$ and $\sigma\in \D(B)$, $\rho\otimes\sigma \in \D(AB)$ represents the tensor product (or Kronecker product) of $\rho$ and $\sigma$. The identity operator on $\H_A$ is denoted $\id_A$. 

Let $\rho_{AB} \in \D(AB)$. We define the {\em partial trace with respect to $A$} of $\rho_{AB}$ as
\[ \rho_{B} \defeq \Tr_{A}(\rho_{AB})
\defeq \sum_{i} (\bra{i} \otimes \id_{\cspace{B}})
\rho_{AB} (\ket{i} \otimes \id_{\cspace{B}}) , \]
where $\set{\ket{i}}_i$ is an orthonormal basis for the Hilbert space $\H_A$.
The state $\rho_B\in \D(B)$ is referred to as a {\em reduced density matrix} or a {\em marginal state}. Unless otherwise stated, a missing register from subscript in a state will represent partial trace over that register. Given a $\rho_A\in\D(A)$, a {\em purification} of $\rho_A$ is a pure state $\rho_{AB}\in \D(AB)$ such that $\Tr_{B}(\rho_{AB})=\rho_A$. Any quantum state has a purification using a register $B$ with $|B|\leq |A|$. The purification of a state, even for a fixed $B$, is not unique as any unitary applied on register $B$ alone does not change $\rho_A$.

An important class of states that we will consider is the {\em classical quantum states}. They are of the form $\rho_{AB} = \sum_a \mu(a) \ketbra{a}_A\otimes \rho^a_B$, where $\mu$ is a probability distribution. In this case, $\rho_A$ can be viewed as a probability distribution and we shall continue to use the notations that we have introduced for probability distribution, for example,  $\expec_{a\leftarrow A}$ to denote the average $\sum_a \mu(a)$. 

A quantum {\em super-operator} (or a {\em quantum channel} or a {\em quantum operation}) $\E: A\rightarrow B$ is a completely positive and trace preserving (CPTP) linear map (mapping states from $\mathcal{D}(A)$ to states in $\mathcal{D}(B)$). The identity operator in Hilbert space $\H_A$ (and associated register $A$) is denoted $\id_A$.  A {\em unitary} operator $\U_A:\H_A \rightarrow \H_A$ is such that $\U_A^{\dagger}\U_A = \U_A \U_A^{\dagger} = \id_A$. The set of all unitary operations on register $A$ is  denoted by $\mathcal{U}(A)$. 

A $2$-outcome quantum measurement is defined by a collection $\{M, \id - M\}$, where $0 \preceq M \preceq \id$ is a positive semidefinite operator, where $A\preceq B$ means $B-A$ is positive semidefinite. Given a quantum state $\rho$, the probability of getting outcome corresponding to $M$ is $\Tr(\rho M)$ and getting outcome corresponding to $\id - M$ is $1-\Tr(\rho M)$. 


\subsubsection{Distance measures for quantum states}

We now define the distance measures we use and some properties of these measures. Before defining the distance measures, we introduce the concept of {\em fidelity} between two states, which is not a distance measure but a similarity measure.

\begin{definition}[Fidelity]
 Let $\rho_A,\sigma_A \in \D(A)$ be quantum states. The fidelity between $\rho$ and $\sigma$ is defined as
$$\F(\rho_A,\sigma_A)\defeq\onenorm{\sqrt{\rho_A}\sqrt{\sigma_A}}.$$
\end{definition}

For two pure states $|\psi\>$ and $|\phi\>$, we have $\F(|\psi\>\<\psi|,|\phi\>\<\phi|) = |\<\psi|\phi\>|$. We now introduce the two distance measures we use.

\begin{definition}[Distance measures]
 Let $\rho_A,\sigma_A \in \D(A)$ be quantum states. We define the following distance measures between these states.
\begin{align*}
\text{Trace distance:}& \quad \Delta(\rho_A,\sigma_A) \defeq \frac{1}{2}\|\rho_A-\sigma_A\|_1  \\
\text{Bures metric:}& \quad \BR(\rho_A,\sigma_A) \defeq \sqrt{1-\F(\rho_A,\sigma_A)}. 
\end{align*}
\end{definition}

Note that for any two quantum states $\rho_A$ and $\sigma_A$, these distance measures lie in $[0,1]$. The distance measures are $0$ if and only if the states are equal, and the distance measures are $1$ if and only if the states have orthogonal support, i.e., if $\rho_A \rho_B = 0$.

Conveniently, these measures are closely related.

\begin{fact}\label{fact:deltabures}For all quantum states $\rho_A,\sigma_A \in \D(A)$, we have
\begin{align*}
&\quad 1-\F(\rho_A,\sigma_A) \leq \Delta(\rho_A,\sigma_A) \leq \sqrt{2} \cdot \BR(\rho_A,\sigma_A). 
\end{align*}
\end{fact}

\begin{proof} The Fuchs-van de Graaf inequalities~\cite{FvdG99,Wat16} state that 
\begin{align*}
&\quad 1-\F(\rho_A,\sigma_A) \leq \Delta(\rho_A,\sigma_A) \leq \sqrt{1-\F^2(\rho_A,\sigma_A)}. 
\end{align*}
Our fact follows from this and the relation $1-\F^2(\rho_A,\sigma_A) \leq 2-2\F(\rho_A,\sigma_A)$.
\end{proof}

A fundamental fact about quantum states is Uhlmann's theorem~\cite{uhlmann76}.
\begin{fact}[Uhlmann's theorem]
\label{fact:uhlmann}
Let $\rho_A,\sigma_A\in \D(A)$. Let $\rho_{AB}\in \D(AB)$ be a purification of $\rho_A$ and $\sigma_{AB}\in\D(AB)$ be a purification of $\sigma_A$ with. There exists a unitary $\U: \H_B \rightarrow \H_B$ such that
 $$\F(\ketbra{\theta}_{AB}, \ketbra{\rho}_{AB}) = \F(\rho_A,\sigma_A) ,$$
 where $\ket{\theta}_{AB} = (\id_A \otimes \U) \ket{\sigma}_{AB}$.
Trivially, the same holds for the Bures metric $\BR$ as well. 
\end{fact}

We now review some properties of the Bures metric that we use in our proofs.

\begin{fact}[Facts about $\BR$]\label{fact:relation-inf} For all quantum states $\rho_A, \rho'_A, \sigma_A, \sigma'_A \in \D(A)$, we have the following.
\end{fact}
\begin{factpart}[Triangle inequality \cite{Bures69}]\label{fact:triangle} The following triangle inequality and a weak triangle inequality hold for the Bures metric and the square of the Bures metric. 
\begin{enumerate}
\item $\BR(\rho_A,\sigma_A) \leq \BR(\rho_A,\tau_A) + \BR(\tau_A,\sigma_A).$
\item $\BR^2(\rho_A^1, \rho_A^{t+1}) \le t \cdot \sum_{i=1}^t \BR^2(\rho_A^i, \rho_A^{i+1}).$
\end{enumerate}
\end{factpart}
\begin{factpart}[Product states]\label{fact:prod} 
$\BR(\rho_A\otimes \sigma_A, \rho'_A \otimes \sigma'_A)   \leq  \BR(\rho_A, \rho'_A) + \BR(\sigma_A, \sigma'_A).$
Additionally, if $\sigma_A=\sigma'_A$ then $\BR(\rho_A\otimes \sigma_A, \rho'_A \otimes \sigma'_A)  = \BR(\rho_A, \rho'_A)$. 
\end{factpart}
\begin{factpart}[Partial measurement]\label{fact:subadd} For classical-quantum states $\theta_{XB},\theta'_{XB}$ with same probability distribution on the classical part, we have
$$\BR^2(\theta_{XB},\theta'_{XB}) = \expec_{x\leftarrow X} [\BR^2(\theta^x_B, \theta'^x_B)].$$
\end{factpart}

\begin{proof}
These facts are proved as follows.
\renewcommand{\theenumi}{\bfseries \Alph{enumi}}
\begin{enumerate}
\item 
Proof of part $2$ follows from triangle inequality and the fact that for positive reals $a_1,a_2,\ldots a_t$, 
$$\Bigl(\sum_i a_i\Bigr)^2 = \sum_i a_i^2+2\sum_{i<j}a_i\cdot a_j \leq \sum_i a_i^2 + \sum_{i<j} \Bigl(a^2_i + a^2_j\Bigr) \leq t\Bigl(\sum_i a_i^2\Bigr).$$

\item 
Follows easily from the triangle inequality. 

\item 
Let $\theta_{XB} = \sum_x p(x) \ketbra{x}\otimes \theta^x_B$ and $\theta'_{XB} = \sum_x p(x) \ketbra{x}\otimes \theta'^x_B$. Then 
\begin{align*}
\F(\theta_{XB},\theta'_{XB})  &=    \Tr\(\sqrt{\sum_x p^2(x) \ketbra{x}\otimes \sqrt{\theta^x_B}\theta'^x_B\sqrt{\theta^x_B}}\)\\ 
&=  \Tr\(\sum_x p(x) \ketbra{x}\otimes \sqrt{\sqrt{\theta^x_B}\theta'^x_B\sqrt{\theta^x_B}}\) \\
&= \sum_x p(x)\F(\theta^x_B,\theta'^x_B) \\
&= \expec_{x\leftarrow X}[\F(\theta^x_B,\theta'^x_B)],
 \end{align*}
which proves the fact.\qedhere
\end{enumerate}
\end{proof}

Finally, an important property of both these distance measures is monotonicity under quantum operations \cite{lindblad75,barnum96}.

\begin{fact}[Monotonicity under quantum operations]
\label{fact:monotonicitydistance}
For quantum states $\rho_A$, $\sigma_A \in \D(A)$, and a quantum operation $\E(\cdot):\L(A)\rightarrow \L(B)$, it holds that
\begin{align*}
	\Delta(\E(\rho) , \E(\sigma)) \leq \Delta(\rho_A,\sigma_A) \quad \mbox{and} \quad \BR(\E(\rho_A),\E(\sigma_A)) \leq \BR(\rho_A,\sigma_A),
\end{align*}
with equality if $\E$ is unitary. 
In particular, for bipartite states $\rho_{AB},\sigma_{AB}\in \D(AB)$, it holds that
\begin{align*}
	\Delta(\rho_{AB},\sigma_{AB}) \geq \Delta(\rho_A,\sigma_A) \quad \mbox{and} \quad \BR(\rho_{AB},\sigma_{AB}) \geq \BR(\rho_A,\sigma_A).
\end{align*}
\end{fact}

\subsubsection{Mutual information and relative entropy}

We start with the following fundamental information theoretic quantities. We refer the reader to the excellent sources for quantum information theory \cite{Wil12, Wat16} for further study.

\begin{definition}\label{def:relentropy}
Let $\rho_A \in \D(A)$ be a quantum state and $\sigma_A \in \D(A)$ be another quantum state on the same space with  $\text{supp}(\rho_A) \subset \text{supp}(\sigma_A)$. We then define the following.
\begin{align*}
\text{von Neumann entropy:}& \quad \textrm{S}(\rho_A) \defeq - \Tr(\rho_A\log\rho_A) .\\
\text{Relative entropy:}& \quad \relent{\rho_A}{\sigma_A} \defeq \Tr(\rho_A\log\rho_A) - \Tr(\rho_A\log\sigma_A) .
\end{align*}
\end{definition}

We now define mutual information and conditional mutual information.

\begin{definition}[Mutual information]
\label{def:entropy}
Let  $\rho_{ABC}\in\D(ABC)$ be a quantum state. We define the following measures.
\begin{align*}
\text{Mutual information:}& \quad \I(A:B)_{\rho}\defeq \ent{\rho_A} + \ent{\rho_B}-\ent{\rho_{AB}} = \relent{\rho_{AB}}{\rho_A\otimes\rho_B} .\\
\text{Conditional mutual information:}& \quad \I(A:B~|~C)_{\rho}\defeq \I(A:BC)_{\rho}-\I(A:C)_{\rho}.
\end{align*}
\end{definition}

We will need the following basic properties. 

\begin{fact}[Properties of S and $\I$]
Let $\rho_{ABC}\in\D(ABC)$ be a quantum state. We have the following.
\end{fact}
\begin{factpart}[Nonnegativity]\label{fact:nonneg}
\begin{align*}
\relent{A}{B}_{\rho}  \geq 0 &\text{ and } |A| \ge \ent{A}_\rho  \geq 0\\
\I(A:B)_{\rho}  \geq 0 &\text{ and }\I(A:B~|~C)_{\rho}  \geq 0.  
\end{align*}
\end{factpart}
\begin{factpart}[Partial measurement]\label{fact:relsubadd} For classical-quantum states, $\theta_{XB},\theta'_{XB}$ with same classical distribution on register $X$:
$$\relent{\theta_{XB}}{\theta'_{XB}} = \expec_{x\leftarrow X} [\relent{\theta^x_B}{\theta'^x_B}].$$
\end{factpart}

\begin{factpart}[Chain rule]\label{fact:chain-rule}
$\I(A:BC)_{\rho} = \I(A:C)_{\rho} + \I(A:B~|~C)_{\rho}  = \I(A:B)_{\rho} + \I(A:C~|~B)_{\rho}.$
\end{factpart}
\begin{factpart}[Monotonicity]\label{fact:mono}  For a quantum operation $\E(\cdot):\L(A)\rightarrow \L(B)$,
$\I(A:\E(B)) \le I(A:B)$ with equality when $\E$ is unitary. In particular $\I(A:BC)_{\rho} \ge  \I(A:B)_{\rho} .$
\end{factpart}
\begin{factpart}[Bar hopping]\label{fact:barhopping}
$\I(A:BC)_{\rho}\geq \I(A:B~|~C)_{\rho}$, where equality holds if $\I(A:C)_{\rho}=0$.
\end{factpart}
\begin{factpart}[Independence]\label{fact:infoind}
If $\I(B:C)_{\rho}=0$, then $\I(A:BC)_{\rho} \geq \I(A:B)_{\rho} + \I(A:C)_{\rho}$.
\end{factpart}
\begin{factpart}[Araki-Lieb inequality]\label{fact:araki}
$|\ent{\rho_{AB}} - \ent{\rho_B}| \le \ent{\rho_A}.$
\end{factpart}
\begin{factpart}[Information bound]\label{fact:inf-bound}
\begin{align*}
\I(A:BC)_{\rho}  \leq  \I(A:C)_{\rho} + 2 \ent{\rho_B}. 
\end{align*}
\end{factpart}
\begin{factpart}[Stronger version of Pinsker's inequality] \label{fact:SvsB} For quantum states $\rho$ and $\sigma$:  
 $$
\relent{\rho}{\sigma} \ge 1 -  \F(\rho, \sigma) = \BR^2(\rho, \sigma) .
$$
\end{factpart}

\begin{factpart}\label{fact:IvsB} For classical-quantum state (register $X$ is classical) $\rho_{XAB}$:
\begin{align*}
\I(A;B|X)_{\rho} &= \mathbb{E}_{x \leftarrow X}  \relent{\rho_{AB}^x}{\rho^x_A \otimes \rho^x_B}  \ge \mathbb{E}_{x \leftarrow X} \BR^2 \left( \rho_{AB}^x, \rho^x_A \otimes \rho^x_B\right) . \\
\I(X;A) & = \relent{\rho_{XA}}{\rho_X\otimes\rho_A}  = \mathbb{E}_{x \leftarrow X}  \relent{\rho^x_A}{\rho_A} . \\
\I(X;A) &= \I(f(X)X;A), \mbox{ where $f$ is any function.}
\end{align*}
\end{factpart}

\begin{proof} These facts are proved as follows.
\renewcommand{\theenumi}{\bfseries \Alph{enumi}}
\begin{enumerate}
\item For nonnegativity of relative entropy, see \cite[Theorem 11.7]{NC00}. For nonnegativity of mutual information and conditional mutual information, see \cite[Theorem 11.6.1]{Wil12} and \cite[Theorem 11.7.1]{Wil12}.
\item 
Let $\theta_{XB} = \sum_x p(x) \ketbra{x}\otimes \theta^x_B$ and $\theta'_{XB} = \sum_x p(x) \ketbra{x}\otimes \theta'^x_B$. Then 
\begin{align*}
\relent{\theta_{XB}}{\theta'_{XB}} 
&=  \sum_x\Tr(p(x)\ketbra{x}\otimes \theta^x_B(\log \theta_{XB} - \log \theta'_{XB}))\\ 
&= \sum_xp(x)\Tr(\theta^x_B(\log(p(x)\theta^x_B) - \log(p(x)\theta'^x_B))) \\
&= \sum_xp(x)\Tr(\theta^x_B(\log\theta^x_B - \log\theta'^x_B))\\
&= \expec_{x\leftarrow X} \relent{\theta^x_B}{\theta'^x_B},
\end{align*}
which proves the fact.
\item 
Follows from direct calculation.
\item See~\cite{NC00} [Theorem $11.15$].

\item 
Follows from Chain rule (\fct{chain-rule}) and Non-negativity (\fct{nonneg}).
\item 
Consider the following relations that use chain rule:
\begin{align*}
\I(A:BC)_{\rho} 
&=  \I(A:B)_{\rho} + \I(A:C~|~B)_{\rho} \\ 
&=  \I(A:B)_{\rho} + \I(AB:C)_{\rho} - \I(B:C)_{\rho}\\ 
&\ge  \I(A:B)_{\rho} + \I(A:C)_{\rho}.
\end{align*}
The last line uses $\I(B:C)_{\rho} =0$ and monotonicity (\fct{mono}).

\item See~\cite{NC00} [Section $11.3.4$]. 

\item Consider,
\begin{align*}
\I(A:BC)_{\rho} &= \I(A:C)_{\rho} + \I(CA:B)_{\rho} - \I(B:C) \\
& \leq  \I(A:C)_{\rho} + \I(CA:B)_{\rho} \\
& \leq  \I(A:C)_{\rho} + \ent{B}  + \ent{CA} - \ent{CAB} \\
& \leq  \I(A:C)_{\rho} + 2 \ent{B} . & \mbox{(\fct{araki})}
\end{align*}

\item 
Using Corollary $4.2$ and Proposition $4.5$ in \cite{Tomamichel16}, we find that 
$$\relent{\rho}{\sigma}\geq -2\log \F(\rho,\sigma).$$ 
The fact now follows since for any positive $x<1$, $2^x > 2\cdot x^2$.

\item 
For the first relation, we proceed as follows, and then use Pinsker's inequality.
\begin{align*}
\I(A:B~|~X)_{\rho} 
& =  \I(A:BX)_{\rho} - \I(A:X)_{\rho} \\
& = \relent{\rho_{ABX}}{\rho_A\otimes\rho_{BX}} - \relent{\rho_{AX}}{\rho_A\otimes \rho_X}\\
& =  \expec_{x\leftarrow X} [\relent{\rho^x_{AB}}{\rho_A\otimes\rho^x_B} - \relent{\rho^x_A}{\rho_A}] \\
& =  \expec_{x\leftarrow X} [-S(\rho^x_{AB}) - \Tr(\rho^x_A\log\rho_A) + S(\rho^x_B)+ S(\rho^x_A)+\Tr(\rho^x_A\log\rho_A)] \\
& =  \expec_{x\leftarrow X} [-S(\rho^x_{AB}) + S(\rho^x_B)+ S(\rho^x_A)] = \expec_{x\leftarrow X} [\relent{\rho^x_{AB}}{\rho^x_A\otimes \rho^x_B}],
\end{align*}
where in third line, we have used Fact \ref{fact:relsubadd}. 
The second relation follows by direct calculation and Fact \ref{fact:subadd}. 
The third relation follows by monotonicity under the maps $\ketbra{x} \rightarrow \ketbra{x}\otimes \ketbra{f(x)}$ and partial trace. \qedhere
\end{enumerate}
\end{proof}

We will need the following relation between $\I$ and $\Delta$ for binary classical-quantum states (see also~\cite{JainN06}).
\begin{claim}\label{clm:I2}
	Let $\rho_{AB} \in \D(AB)$ be a classical quantum state of the form $\rho_{AB} = p\ketbra{0}_A\otimes \rho^0_B+(1-p)\ketbra{1}\otimes\rho^1_B$ . Then
	\[\I(A:B)_{\rho}\leq  2\log(2) \cdot  \Delta(p\rho^0_B,(1-p)\rho^1_B).\] 
\end{claim}

\begin{proof}
We drop the register index from $\rho^0_B,\rho^1_B$. Let $\rho_{av}=p\rho^0+(1-p)\rho^1$. Consider
\begin{align*}
\I(A:B)_{\rho} 
&= p\relent{\rho^0}{\rho_{av}} + (1-p)\relent{\rho^1}{\rho_{av}} & (\text{\fct{IvsB}})\\ 
&= \relent{p\rho^0}{\frac{1}{2}\rho_{av}} + \relent{(1-p)\rho^1}{\frac{1}{2}\rho_{av}}- p\text{log}(2) - (1-p)\text{log}(2) + \ent{p}\\ 
&\leq \relent{p\rho^0}{\frac{1}{2}\rho_{av}} + \relent{(1-p)\rho^1}{\frac{1}{2}\rho_{av}}.
\end{align*}

The last inequality follows from $\ent{p} \leq \log(2)$.
Now, using ~\cite[Theorem~9]{Audenaert14}, which states that $$\relent{p\rho^0}{\frac{1}{2}\rho_{av}} \leq \log(2)\Delta(p\rho^0,(1-p)\rho^1)\quad \text{and}\quad \relent{(1-p)\rho^1}{\frac{1}{2}\rho_{av}}\leq \log(2)\Delta(p\rho^0,(1-p)\rho^1),$$ the claim follows.
\end{proof}

Our next claim gives us a way to use high mutual information between two registers in a classical quantum state to make a prediction about the classical part using measurement on the quantum part.

\begin{claim}[Information $\Rightarrow$ prediction] \label{clm:maxlike}
Let $\rho_{AB} \in \D(AB)$ be a classical quantum state of the form $\rho_{AB} = p\ketbra{0}_A\otimes \rho^0_B+(1-p)\ketbra{1}\otimes\rho^1_B$ .
The probability of predicting $A$ by a measurement on $B$ is  at least
\[\frac{1}{2}+\frac{\I(A:B)}{2\log 2}.\]
\end{claim}
\begin{proof}
We drop the register label $B$. Let $M$ be a projector on the support of positive eigenvectors of the state $p\rho^0-(1-p)\rho^1$. Let the measurement be $\{M,\id-M\}$ and first outcome imply $0$ in register $A$ and second outcome imply $1$. Then probability of success is
\begin{align*}
p\Tr(\rho^0M) + (1-p)\Tr(\rho^1(\id-M)) 
&= (1-p)+\Tr((p\rho^0 - (1-p)\rho^1)M) & \\
&= (1-p) + \frac{1}{2}(\|p\rho^0 - (1-p)\rho^1\|_1 + \Tr(p\rho^0 - (1-p)\rho^1)) & \\
&= (1-p) + \frac{1}{2}(\|p\rho^0 - (1-p)\rho^1\|_1 + 2p-1) & \\
&=\frac{1}{2} + \frac{1}{2}\|p\rho^0 - (1-p)\rho^1\|_1 & \\
&=\frac{1}{2} + \Delta(p\rho^0,(1-p)\rho^1)  .
\end{align*} 
From \clm{I2}, we know that $\Delta(p\rho^0,(1-p)\rho^1) \geq {\I(A:B)}/(2\log 2)$. 
\end{proof}

\subsection{Quantum Communication complexity}
\label{sec:comm}

In quantum communication complexity, two players wish to compute a classical function $F\colon\X \times \Y \to \{0,1\}$ for some finite sets $\X$ and $\Y$. The inputs $x\in \X$ and $y\in\Y$ are given to two players Alice and Bob, and the goal is to minimize the quantum communication between them required to compute the function. 

While the players have classical inputs, the players are allowed to exchange quantum messages. Depending on whether or not we allow the players arbitrary shared entanglement, we get $\mathrm{Q}(F)$, bounded-error quantum communication complexity without shared enganglement and $\Q(F)$, for the same measure with shared entanglement. Obviously $\Q(F) \leq \mathrm{Q}(F)$.
In this paper we will only work with $\Q(F)$, which makes our results stronger since we prove lower bounds in this work. 

Let $F\colon \X \times \Y \rightarrow \{0,1,*\}$ be a partial function, with $\dom(F) \defeq \{(x,y)\in \X \times \Y: F(x,y) \neq *\}$, and let $\eps \in (0,1/2)$.

 An entanglement assisted quantum communication protocol $\Pi$ for this function is as follows. Alice and Bob start with a preshared entanglement. Upon receiving inputs $(x,y)$, where Alice gets $x$ and Bob gets $y$, they exchange quantum states and then Alice applies a measurement on her qubits to output $1$ or $0$. Let $O(x,y)$ be the random variable output by Alice in $\Pi$, given input $(x,y)$. Let $\mu$ be a distribution over $\dom(F)$. 

Let inputs to Alice and Bob be given in registers $X$ and $Y$ in the state $$\sum_{x,y}\mu(x,y)\ketbra{x}_X\otimes\ketbra{y}_Y.$$ Let these registers be purified by $R_X$ and $R_Y$ respectively, which are not accessible to either players. Let Alice and Bob initially hold register $A_0,B_0$ with shared entanglement $\Theta_{0,A_0B_0}$. Then the initial state is 
\begin{eqnarray*}
\label{eq:initialstate}
\ket{\Psi_0}_{XYR_XR_YA_0B_0} \defeq \sum_{x,y}\sqrt{\mu(x,y)}\ket{xxyy}_{XR_XYR_Y}\ket{\Theta_0}_{A_0B_0}
\end{eqnarray*}

Alice applies a unitary $U^1: XA_0\rightarrow XA_1C_1$ such that the unitary acts on $A_0$ conditioned on $X$.  She sends $C_1$ to Bob. Let $B_1\equiv B_0$ be a relabelling of Bob's register $B_0$. He applies $U^2: YC_1B_1\rightarrow YC_2B_2$ such that the unitary acts on $C_1B_0$ conditioned on $Y$. He sends $C_2$ to Alice. Players proceed in this fashion till end of the protocol. At any round $r$, let the registers be $A_rC_rB_r$, where $C_r$ is the message register, $A_r$ is Alice's register and $B_r$ is Bob's register. If $r$ is odd, then $B_r \equiv B_{r-1}$ and if $r$ is even, then $A_r\equiv A_{r-1}$. Let the joint state in registers $A_rC_rB_r$ be $\Theta_{r,A_rC_rB_r}$. Then the global state at round $r$ is
\begin{eqnarray*}
\label{eq:roundrstate}
\ket{\Psi_r}_{XYR_XR_YA_rC_rB_r} \defeq \sum_{x,y}\sqrt{\mu(x,y)}\ket{xxyy}_{XR_XYR_Y}\ket{\Theta_r}_{A_rC_rB_r}.
\end{eqnarray*}

We define the following quantities.
\begin{align*}
\textrm{Worst-case error:} & \quad \err(\Pi) \defeq \max_{(x,y)\in \dom(F)} \{ \Pr[O(x,y) \neq F(x,y)] \} .\\
\textrm{Distributional error:} &\quad  \err^\mu(\Pi) \defeq   \expec_{(x,y) \leftarrow \mu} \Pr[O(x,y) \neq F(x,y)]. \\ 
\textrm{Quantum CC of a protocol:}&\quad  \QCC(\Pi) \defeq \sum_i |C_i| . \\
\textrm{Quantum CC of $F$:}&\quad  \Q_\eps(F) \defeq \min_{\Pi: \err(\Pi) \leq \eps} \QCC(\Pi).  \\
\end{align*}


Our first fact justifies using $\eps=1/3$ by default since the exact constant does not matter since the success probability of a protocol can be boosted for QCC.

\begin{fact}[Error reduction] \label{fact:boost}
Let $0 < \delta < \eps < 1/2$. Let $\Pi$ be a protocol for $F$ with $\err(\Pi) \leq \eps$. There exists protocol $\Pi'$ for $F$ such that $\err(\Pi') \leq \delta$ and   
\begin{align*}
\QCC(\Pi') \leq O\left(\frac{\log(1/\delta)}{\big(\frac{1}{2}-\eps\big)^2}\cdot  \QCC(\Pi)\right). 
\end{align*}
\end{fact}

This fact is proved by simply repeating the protocol sufficiently many times and taking the majority vote of the outputs. If the error $\epsilon$ is close to $1/2$, we can first reduce the error to a constant by using $O(\frac{1}{(1/2-\eps)^2})$ repetitions. Then $O\left(\log(1/\delta)\right)$ repetitions suffice to reduce the error down to $\delta$. Hence the quantum communication only increases by a factor of $O\left(\frac{\log(1/\delta)}{(1/2-\eps)^2}\right)$.

We have the following relation between worst-case and average-case error quantum communication complexities. It follows for example from standard application of Sion's minimax theorem~\cite{Sion58}.
\begin{fact}[Minimax principle] \label{fact:equiv}
Let $F\colon\X \times \Y \rightarrow \{0,1,*\}$ be a partial function. Fix an error parameter $\eps \in (0,1/2)$ and a quantum communication bound $q \geq 0$. Suppose $\mathcal{F}$ is a family of protocols such that for every distribution $\mu$ on $\dom(F)$ there exists a protocol $\Pi \in \mathcal{F}$ such that $$\err^\mu(\Pi) \leq \eps \quad \mathrm{and} \quad \QCC(\Pi)\leq q.$$
Then there exists a protocol $\Pi'$ such that $$\err(\Pi')\leq \eps \quad \mathrm{and} \quad \QCC(\Pi') \leq q.$$
\end{fact}

Our next claim shows that having some information about the output of a Boolean function $F$ allows us to predict the output of $F$ with some probability greater than $1/2$.

\begin{claim} \label{clm:lowinf}
Let $F\colon\X \times \Y \rightarrow \{0,1,*\}$ be a partial function and $\mu$ be a distribution over $\dom(F)$. Let $XY$ be registers with the state $\sum_{x,y}\mu(x,y)\ketbra{x}\otimes\ketbra{y}$ and define a register $F$ that contains the value of $F(x,y)$. Let $\Pi$ be a quantum communication protocol with registers $X,Y$ input to Alice and Bob respectively and number of rounds $r$ (which is even). There either
\begin{itemize}
\item There exists a quantum communication protocol  $\Pi'$ for $F$ with $r$ rounds, with input $(X,Y)$ to Alice and Bob respectively, such that 
$$\QCC(\Pi') = \QCC(\Pi) + 1, \quad\text{and}\quad \err^\mu(\Pi') < \frac{1}{2} - \frac{\I(F:A_rC_r~|~X)_{\Psi_r}}{2\log(2)} . $$
\item Or, there exists a  quantum communication protocol  $\Pi'$ for $F$ with $r$ rounds, with input $(X,Y)$ to Alice and Bob respectively, such that 
$$\QCC(\Pi') \leq \QCC(\Pi), \quad\text{and}\quad \err^\mu(\Pi') < \frac{1}{2} - \frac{\I(F:B_{r}C_r~|~Y)_{\Psi_r}}{2\log(2)} . $$
\end{itemize}
\end{claim}
\begin{proof}

We first prove the first case. In $\Pi'$, Alice and Bob run the protocol $\Pi$, after which Alice proceeds as follows. Consider the state $\Psi_{r,XFA_rC_r}$ in registers $XFA_rC_r$ (note that we have added a new register $F$ to the state $\Psi_r$, which can be done naturally). Let $$\Psi_{r,XFA_rC_r} = \sum_x \mu(x)\ketbra{x}_X\otimes \Psi^x_{r,FA_rC_r}$$ be the decomposition of $\Psi_{r,XFA_rC_r}$, which is possible since $X$ is classical. Note that $\Psi^x_{r,FA_rC_r}$ is a classical quantum state between the registers $F$ and $A_rC_r$.  Alice,  essentially applying \clm{maxlike} makes a prediction about the content of register $F$. Then she outputs the prediction. Clearly,
\[\QCC(\Pi') = \QCC(\Pi) + 1 . \]
For every input $x$ for Alice, her prediction is successful with probability at least ${1}/{2}+{\I(F:A_rC_r)_{\Psi^x_r}}/{2\log(2)}$
by \clm{maxlike}. Hence the overall success probability of $\Pi'$ is at least
\[\expec_{x\leftarrow X}\left[\frac{1}{2}+\frac{\I(F:A_rC_r)_{\Psi^x_r}}{2\log(2)}\right]
=\frac{1}{2}+\frac{\I(F:A_rC_r|X)_{\Psi_r}}{2\log(2)}.\qedhere \]

Second case follows with same argument, but applied on Bob' side before he sends $C_r$ to Alice. Bob then sends the outcome to Alice instead of $C_r$. 
\end{proof}

The following claim is used in our proof to handle the easy case of a biased input distribution. 

\begin{claim} \label{clm:biasmu}
Let $F\colon\X \times \Y \rightarrow \{0,1,*\}$ be a partial function and let $\mu$ be a distribution over $\dom(F)$.  Let $\epsilon \in (0,1/2)$ and $c \geq 1$ be a positive integer.    For $i \in [c]$, let $X_i,Y_i$ be registers with the state $\sum_{x,y}\mu(x,y)\ketbra{x}_{X_i}\otimes\ketbra{y}_{Y_i}$ and define register $L_i$ that holds the value $F(x_i,y_i)$. Define $X \defeq X_1\ldots X_c$, $Y \defeq Y_1\ldots Y_c$, and $L \defeq L_1\ldots L_c$. Let $\Psi_{XYL}$ be the joint state in registers $X,Y,L$. Then 
either
\begin{enumerate}
\item[(a)] There exists a protocol $\Pi$ for $F$ such that $\QCC(\Pi) = 1$, and $\err^\mu(\Pi) \leq \frac{1}{2} - \epsilon$, or
\item[(b)] $\Delta(\Psi_{XL}, \Psi_X \otimes W_{L_1}\otimes\ldots W_{L_c})  \leq c\epsilon $,
where $W_{L_i}$ is the maximally mixed state in register $L_i$.
\end{enumerate}
\end{claim}
\begin{proof}
Define, $q^{x_1} \defeq \Pr[F=0~|~X_1=x_1]$. Assume $\expec_{x_1\leftarrow X_1} \left[\left| \frac{1}{2} - q^{x_1}\right| \right] \geq \epsilon$.  Let $\Pi$ be a protocol where Alice, on input $x_1$, outputs $0$ if $q^{x_1} \geq 1/2$ and $1$ otherwise. Then, 
$$\err^\mu(\Pi)  =  \frac{1}{2} - \expec_{x_1\leftarrow X_1}\left| \frac{1}{2} - q^{x_1}\right| \leq \frac{1}{2} - \epsilon. $$
Assume otherwise $\expec_{x_1\leftarrow X_1}\left| \frac{1}{2} - q^{x_1}\right|  < \epsilon$. This implies
\begin{align*}
 \Delta(\Psi_{XL}, \Psi_X \otimes W_{L_1}\otimes\ldots W_{L_c}) &\leq c \cdot \Delta(\Psi_{X_1L_1}, \Psi_{X_1} \otimes W_{L_1}) = c \cdot  \expec_{x_1\leftarrow X_1}\left| \frac{1}{2} - q^{x_1}\right| < c\epsilon,
\end{align*}
where the first inequality follows from \fct{prod}.
\end{proof}

In below, let $A'_r, B'_r$ represent Alice and Bob's registers at round $r$. That is, at even round $r$, $A'_r = A_rC_r, B'_r=B_r$ and at odd $r$, $A'_r=A_r, B'_r=B_rC_r$. We will need the following version of quantum-cut-and-paste lemma from \cite{NT16} (also see \cite{JRS03} for a similar argument, where it is used to lower bound quantum communication complexity of disjointness). This is a special case of \cite[Lemma 7]{NT16} and we have rephrased it using our notation. 

\begin{lemma} [Quantum cut-and-paste] \label{lem:quantum_cut_paste} Let $\Pi$ be a quantum protocol with classical inputs and consider distinct inputs $u,u'$ for Alice and $v,v'$ for Bob. Let $\ket{\Psi_{0,A_0 B_0}}$ be the initial shared state between Alice and Bob. Also let $\ket{\Psi_{k,A'_k B'_k}^{u'',v''}}$ be the shared state after round $k$ of the protocol when the inputs to Alice and Bob are $(u'',v'')$ respectively. For $k$ odd, let
$$
h_k = \BR \left( \Psi_{k,B'_k}^{u,v}, \Psi_{k,B'_k}^{u',v}\right)
$$
and for even $k$, let
$$
h_k = \BR \left( \Psi_{k,A'_k}^{u,v}, \Psi_{k,A'_k}^{u,v'}\right) .
$$
Then 
$$
\BR \left( \Psi_{r,A'_r}^{u',v},  \Psi_{r,A'_r}^{u',v'}\right) \le h_r + h_{r-1} + 2 \sum_{k=1}^{r-2} h_k .
$$
\end{lemma} 

The following lemma (see also~\cite{CleveDNT97}) formalizes the following intuition: In a quantum protocol with communication $q$, the amount of information that Bob has about Alice's input at any time point is at most $2q$ (note that the factor of $2$ is necessary because of super-dense coding.).

\begin{lemma}\label{lem:qic<=2qcc} Let $\Pi$ be a quantum protocol with the inputs of Alice and Bob $(X,Y)$ being jointly distributed. Alice has an additional input $U$ which is independent of both $(X,Y)$. Let $\mu$ denote the distribution of inputs so that $\mu(x,u,y) = \mu(x,y) \mu(u)$. Let the total pure state after the $k^{\text{th}}$ round of the protocol be 
$$
\ket{\Psi_{k}}_{X \tX Y \tY A'_k B'_k} = \sum_{x,y} \sqrt{\mu(x,y) \mu(u)} \ket{xxuu}_{X \tX U \tU} \ket{yy}_{Y \tY} \ket{\Theta^{x,u, y}_{k}}_{A'_k B'_k} .
$$
Then 
$$
\I(B'_k Y \tY: U|X)_{\Psi_k} \le 2q_k .
$$
Here $q_k$ is communication cost up to round $k$. A similar statement holds by reversing the roles of Alice and Bob. 
\end{lemma}

\begin{proof}
We prove the first inequality by induction on $k$.  The inequality holds trivially for $k=0$. First suppose $k$ is even, so that Bob sent the last message. Then,
\begin{align*}
\I(B'_k Y \tY: U|X)_{\Psi_k} &\le \I(B'_{k-1} Y \tY: U|X)_{\Psi_{k-1}} & \mbox{(\fct{mono})} \\
&\le 2q_{k-1} \le 2 q_k,
\end{align*}

where the first inequality follows by induction step.

Now suppose $k$ is odd, so that Alice sent the last message. By our notation, $B'_k \equiv C_k B_k$ where $C_k$ is Alice's message. Then,
\begin{align*}
\I(B'_{k} Y \tY: U|X)_{\Psi_{k}} &= \I(C_{k}B_{k}Y \tY: U|X)_{\Psi_{k}} \\
& \le \I(B_{k}Y \tY: U|X)_{\Psi_{k}} + 2\ent{C_k|X} & \mbox{(\fct{inf-bound})}\\
& = \I(B'_{k-1}Y \tY: U|X)_{\Psi_{k-1}} + 2 \ent{C_k|X} \\
& \le 2 q_{k-1} + 2 |C_k| = 2 q_k ,
\end{align*}
where last inequality follows from induction step.
\end{proof}
 
\section{Separation}
\label{sec:separation}

In this section we establish the main result, a nearly quadratic separation between quantum communication complexity and the logarithm of approximate rank, which we restate below.

\separation*


Our proof is organized as follows. In \sec{lookup} we define lookup functions, which we will use to construct the function achieving the separation in \thm{separation}. Then in \sec{sep} we prove \thm{separation} using results from later sections. More precisely, we prove the upper bound on our function's approximate rank using \thm{upper}, proved in \sec{upper}. We prove the lower bound using \cor{lower}, which follows from \thm{qcs} in \sec{lower}.
\thm{upper} and \cor{lower} provide a black-box way of using the results of \sec{upper} and \sec{lower} without delving into their proofs. 


\subsection{Lookup functions}
\label{sec:lookup}

We define a simpler version of lookup functions than the ones used in \cite{ABB+16a}, since we only deal with total functions in this paper. This is only for simplicity, and the lower bound shown in this paper also applies to the more general lookup functions for partial functions defined in \cite{ABB+16a}.

First, for any function $F\colon\X \times \Y \to \B$ and integer $c>0$, we can define a new function $F^c \colon \X^c \times \Y^c \rightarrow \B^c$ as
$F^c((x_1, \ldots, x_c), (y_1, \ldots, y_c)) = (F(x_1, y_1), \ldots, F(x_c,y_c))$, which takes $c$ inputs to $F$ and outputs the answers to all $c$ inputs. $F^c$ is simply the problem of computing $F$ on $c$ independent inputs and outputting all $c$ answers.

An $(F, \mathcal{G})$-lookup function, denoted $F_\mathcal{G}$, is defined by a  function $F\colon \X \times \Y \to \{0,1\}$ and a family $\mathcal{G}=\{G_0, \ldots, G_{2^c-1}\}$ of functions, where each $G_i \colon (\X^c \times \{0,1\}^m) \times (\Y^c \times \{0,1\}^m) \to \B$.  It can be viewed as a generalization of the address function.  Alice receives input $\x=(x_1, \ldots, x_c) \in \X^c$ and $\fu=(u_0, \ldots, u_{2^c-1}) \in \{0,1\}^{m 2^c}$ 
and likewise Bob receives input $\y=(y_1, \ldots, y_c) \in \Y^c$ and $\fv=(v_0, \ldots, v_{2^c-1}) \in \{0,1\}^{m2^c}$. We refer to the inputs $(\x,\y)$ as the ``address part'' of the input and the inputs $(\fu,\fv)$ as the ``array part'' of the input.
We will refer to $u_i$ and $v_i$ as a ``cell'' of the array.
The \emph{address}, $\ell$, is determined by the evaluation of $F$ on $(x_1, y_1), \ldots, (x_c,y_c)$, that is 
 $\ell = F^c(\x,\y) \in \{0,1\}^c$.  This address (interpreted as an integer in $\{0,\ldots, 2^c-1\}$) then determines which function, out of the $2^c$ functions $G_i$, the  players should evaluate and which pair of cells, out of the $2^c$ possible pairs $(u_i,v_i)$, of the array are relevant to the output of the function.  The goal of the players is to output 
 $G_\ell(\x, u_\ell, \y, v_\ell)$. The formal definition is the following.

\begin{definition}[$(F, \G)$-lookup function for total $F$]
\label{def:lookup}
Let $F\colon \X \times \Y \to \{0,1\}$ be a function and $\mathcal{G}=\{G_0, \ldots, G_{2^c-1}\}$ a 
family of functions, where each $G_i \colon (\X^c \times \{0,1\}^m) \times (\Y^c \times \{0,1\}^m) \to \B$.  
An $(F, \mathcal{G})$-lookup function, denoted $F_\mathcal{G}$, is a function 
$$F_\G \colon (\X^c \times \B^{m2^c}) \times \Y^c \times \B^{m2^c}\to \B$$
defined as follows.
Let $\x=(x_1, \ldots, x_c) \in \X^c$, $\y =(y_1, \ldots, y_c) \in \Y^c$, $
\fu=(u_0, \ldots, u_{2^c-1}) \in \B^{m2^c}$, and $\fv=(v_0, \ldots, v_{2^c-1}) \in \B^{m2^c}$.  Then
\[
F_{\G}(\x, \fu, \y, \fv) = 
G_\ell(\mathbf{x}, u_\ell, \mathbf{y}, v_\ell), 
\]
where $\ell = F^c(\x,\y)$. 
\end{definition}

Since we only deal with total functions $F$, we will not need to impose a consistency condition for instances where some input to $F$ is outside its domain. (In \cite{ABB+16a}, this condition was called ``consistency outside $F$.'')

In order to show lower bounds on the communication complexity of $F_\G$ (\thm{qcs}) we add two constraints on the family $\G$ as in \cite{ABB+16a}.  

\begin{definition}[Nontrivial XOR family]
\label{def:xor}
Let $\mathcal{G}=\{G_0, \ldots, G_{2^c-1}\}$ a 
family of communication functions, where each $G_i \colon (\X^c \times \{0,1\}^m) \times (\Y^c \times \{0,1\}^m) \to \B$. We 
say that $\G$ is a nontrivial XOR family if the following conditions hold.
\begin{enumerate}
\item (Nontriviality) For all $\x=(x_1, \ldots, x_c) \in \X^c$ and $\y =(y_1, \ldots, y_c) \in \Y^c$, if we have 
$\ell = F^c(\x,\y) \in \B^c$ 
then there exist $u,v,u', v' \in \B^m$ such that 
$G_\ell(\x, u, \y, v) \ne G_\ell(\x, u', \y, v')$.
\item (XOR function) For all $i \in \{0,\ldots,2^c-1\}, u,u',v,v' \in \B^m$ and $\x=(x_1, \ldots, x_c) \in \X^c, 
\y=(y_1, \ldots, y_c) \in \Y^c$ 
if $u \oplus v = u' \oplus v'$ then $G_i(\x, u, \y, v) = G_i(\x, u', \y, v')$.
\end{enumerate}
\end{definition}

The first condition simply enforces that the content of the correct part of the array, i.e., $(u_\ell,v_\ell)$, is relevant to the output of the function in the sense that there is some setting of these bits that makes the function true and another setting that makes it false.

The second condition enforces that the output of the function only depends on $u_\ell \oplus v_\ell$, and not $u_\ell$ and $v_\ell$ individually. This is just one way of combining the arrays of Alice and Bob to form one virtual array that contains $2^c$ cells. Other combining functions are also possible.

\subsection{Separation}
\label{sec:sep}

We can now prove the separation using results from \sec{upper} and \sec{lower}.
Our proof strategy is depicted in \fig{separation}.  

\setlength\tabcolsep{2ex}
\def\arraystretch{2}
\begin{figure}
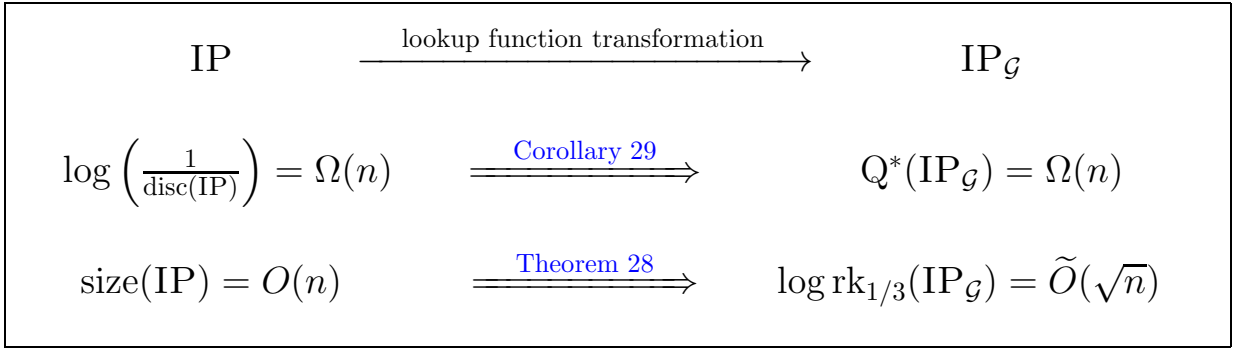

\centering \Large
\begin{tabular}{|c c c|}
\hline
\makebox[7em]{$\IP$} & \hspace{-6ex} $\xrightarrow{\quad\textrm{lookup function transformation}\quad}$ \hspace{-6ex} & \makebox[7em]{$\IP_\G$}\\[1ex]
\hspace{2ex}$\log\left(\frac{1}{\mathrm{disc}(\IP)}\right) = \Omega(n)$ & $\xRightarrow{\quad\textrm{\cor{lower}}\quad}$ & $\Q(\IP_\G) = \Omega(n)$\\[1ex]
$\Size(\IP) = O(n)$ & $\xRightarrow{\quad\textrm{\thm{upper}}\quad}$ & $\log \arank(\IP_\G) = \tO(\sqrt{n})$ \hspace{2ex} \\[1.5ex]
\hline
\end{tabular}
\caption{High-level overview of our separation. Here $\IP:\B^n\times \B^n \to \B$ is the inner product function, $\disc$ is the discrepancy, and $\Size$ is the
circuit size.
\label{fig:separation}}
\end{figure}

The separating function is going to be a lookup function $F_\mathcal{G}$ defined by a 
function $F\colon \X \times \Y \to \{0,1\}$ and a function family $\mathcal{G}=\{G_0, \ldots, G_{2^c-1}\}$.
We will choose $F$ to be the well-known inner product function $\IP:\B^n \times \B^n \to \B$ defined as
\[\IP(x,y) = \mathop{\bigoplus}_{i=1}^n (x_i \wedge y_i).\]

The communication complexity of the inner product function is well understood and is $\Theta(n)$ in all the models discussed in this paper. In fact, even $\log \textrm{sign-rank}(F) = \Theta(n)$~\cite{For02}, where $\textrm{sign-rank}(F)$ is defined as the minimum rank of a matrix $G$ such that $\ell_\infty(F-G)<1/2$. 

To define our function family $\G$, we use the following theorem proved in \sec{upper}.

\begin{restatable}{theorem}{thmupper}\label{thm:upper}
	Let $F$ be a total function with circuit size $\Size(F)$.
	Then for all $c>0$, there exists a nontrivial family of XOR functions $\mathcal{G}=\{G_0,G_1,\ldots,G_{2^c-1}\}$, such that
	\[\log \arank(F_\G) = \tO(c^{3/2} \sqrt{\Size(F)}).\]
\end{restatable}

This theorem gives us a function family $\G$ and proves that for this family we have
\begin{equation}
\label{eq:upper}
\log \arank(\IP_\G) = \tO(c^{3/2} \sqrt{\Size(\IP)}) = \tO(c^{3/2} \sqrt{n}),
\end{equation}
where we use the fact that $\Size(\IP) = O(n)$. This follows because $\IP$ is a parity of size $n$ composed with an $\AND$ function on two bits, and has a circuit of size $O(n)$ consisting of a $\log n$-depth tree of fanin-2 $\XOR$ gates with fanin-2 $\AND$ gates at the bottom.

To show the lower bound, we use the following corollary of \thm{qcs}.

\begin{restatable}{corollary}{corlower}\label{cor:lower}
Let $F_\G$ be an $(F,\G)$-lookup function for a function $F$ and a nontrivial family of XOR functions $\mathcal{G}=\{G_0,G_1,\ldots,G_{2^c-1}\}$ with $c = \Theta (\log(\Q(F)))$. Then
\[\Q(F_\G) = \Omega(\log({1}/{\mathrm{disc}(F)})).\]
\end{restatable}

Here $\mathrm{disc}(F)$ is the discrepancy of $F$ (\defn{disc}). Since $\log({1}/{\mathrm{disc}(\IP)})=\Omega(n)$~\cite[Example 3.19]{KN06}, using \thm{qcs} we have
\begin{equation}
\label{eq:lower}
\Q(\IP_\G) = \Omega(\log({1}/{\mathrm{disc}(\IP)}))= \Omega(n).
\end{equation}
We can now choose $c=\Theta(\log n)$ to satisfy the conditions of \cor{lower}. Thus \eq{upper} yields
\begin{equation*}
\log \arank(\IP_\G) = \tO(\sqrt{n}),
\end{equation*}
which together with \eq{lower} gives us $\Q(\IP_\G)=\tOmega(\log^2(\arank(\IP_\G)))$, proving \thm{separation}.



\section{Upper bound on approximate rank of lookup functions}
\label{sec:upper}

The aim of this section is to prove \thm{upper}.
\thmupper*
Proving this will require some work and we will  need to carefully choose our function family $\mathcal{G}=\{G_0, \ldots, G_{2^c-1}\}$. To do this, we first introduce the concept of an \emph{unambiguous} lookup function.

\begin{definition}
\label{def:unamb}
	Let $F_\mathcal{G}$ be an $(F,\mathcal{G})$-lookup function for a function $F:\X \times \Y \to \B$ and a
	function family $\mathcal{G}=\{G_0,G_1,\ldots,G_{2^c-1}\}$. We say that $F_\mathcal{G}$ is an \emph{unambiguous} lookup function if $G_\ell$ evaluating to $1$ certifies that $F^c(\x,\y) = \ell$. That is, for all $\mathbf{x}, u, \mathbf{y}, v$, $G_\ell(\mathbf{x}, u, \mathbf{y}, v) = 1\Rightarrow F^c(\x, \y)=\ell$.
\end{definition}

Note that not all lookup functions are unambiguous even if we enforce the nontrivial XOR family condition (\defn{xor}), since the condition for when $G_i$ evaluates to $1$ need not even depend on $\x$ and $\y$. For example, $G_i(\x,u,\y,v)$ could simply be some nonconstant function of the string $u \oplus v$.
However, the condition of unambiguity is quite natural, and the lookup functions used in prior work are unambiguous lookup functions (or can be slightly modified to be unambiguous).

The advantage of unambiguous lookup functions is that we can upper bound their approximate rank as follows.

\begin{lemma}\label{lem:logrank_cheatsheet}
	Let $F_\G$ be an unambiguous $(F,\mathcal{G})$-lookup function.
	Then we have 
$$\log \arank(F_\G) = O(c\cdot \max_i \Q(G_i)).$$
\end{lemma}

\begin{proof}
We start by observing that the unambiguity condition implies that for any input $(\x, \fu, \y, \fv)$, at most one of the functions $G_i(\x, u_i, \y, v_i)$ equals $1$. Indeed, only $G_\ell(\x, u_\ell, \y, v_\ell)$ can potentially evaluate to $1$, where $\ell = F^c(\x,\y)$.

In other words, when $F_\G(\x, \fu, \y, \fv)=1$ we must have
$G_\ell(\x,u_\ell,\y,v_\ell)=1$ for $\ell=F^c(x,y)$ and
$G_i(\x,u_i,\y,v_i)=0$ for all $i\neq \ell$. 
On the other hand, when $F_\G(\x, \fu, \y, \fv)=0$ we must have $G_i(\x,u_i,\y,v_i)=0$ for all $i\in\{0,\ldots,2^c-1\}$.

This means the communication matrix of $F_\G$ equals the sum of the communication matrices of $G_i$ over all $i$. More precisely, we extend the definition of $G_i$ to have it take all of $(\x, \fu, \y, \fv)$ as input in the natural way (i.e., it ignores all the other cells of the array except $u_i$ and $v_i$).
This observation directly yields 
$$\rank(F_\G) \leq \sum_{i=0}^{2^c-1} \rank(G_i).$$

The same inequality does not immediately hold for approximate rank, because the errors in the approximation can add up. So even though $A = \sum_i B_i$, if $\tilde{B}_i$ satisfies $\ell_\infty(\tilde{B}_i-B_i) \leq 1/3$, it is not necessarily the case that $\ell_\infty(A - \sum_i \tilde{B}_i) \leq 1/3$. However, if  each $\tilde{B}_i$ is an excellent approximation to $B_i$, then their sum will still be a good approximation to $A$. 
More precisely, it is still the case that 
$$\arank(F_\G) \leq \sum_{i=0}^{2^c-1} \rank_{\eps}(G_i),$$
where $\eps \leq  2^{-c}/3$, since the definition of approximate rank allows error at most $1/3$. This yields
$$\rank(F_\G) \leq 2^c \max_i \rank_{\eps}(G_i) \implies \log \arank(F_\G) \leq c + \max_i \log \rank_{\eps}(G_i).$$
Since log of approximate rank lower bounds quantum communication
complexity, we have that $\log \rank_{\eps}(G_i) \leq \mathrm{Q}^{*}_\eps(G_i)$. By using standard error reduction, we have that $\mathrm{Q}^{*}_\eps(G_i)$ for $\eps=2^{-c}/3$ is at most $O(c \, \Q(G_i))$.
Hence $\log \arank(F_\G)=O(c\cdot \max_i\Q(G_i)).$
\end{proof}

To prove \thm{upper}, we need a tool for taking a function $F$
and finding a collection $\mathcal{G}$ such that $F_{\mathcal{G}}$
is an unambiguous lookup function, and $\Q(G_i)$ is small
for all $G_i\in\mathcal{G}$. The following lemma
provides such a tool.

\begin{lemma}\label{lem:show_your_work}
	Let $F:\B^n\times\B^n\to\B$ be a total function with circuit size $\Size(F)$ (i.e., $F$ can be computed by a Boolean circuit with $\Size(F)$ gates of constant fanin).

	Then for all $c>0$, there exists a nontrivial family of XOR functions $\mathcal{G}=\{G_0,G_1,\ldots,G_{2^c-1}\}$, 
	such that $F_{\mathcal{G}}$ is an unambiguous lookup function and for all $i\in \{0,\ldots,2^c-1\}$, $$\Q(G_i)=\tO(\sqrt{c\,\Size(F)}).$$
\end{lemma}

\begin{proof}
We need to construct functions $G_i(\x,u,\y,v)$ that lead to an unambiguous lookup function (\defn{unamb}), that are a nontrivial XOR family (\defn{xor}) and have $\Q(G_i)=\tO(\sqrt{c\,\Size(F)})$.

Each $G_i$ will check that $u_i\oplus v_i$ has a very special
type of certificate that proves that $F^c(x,y)=i$.
If it contains such a certificate, $G_i$ will output $1$ and otherwise it will output $0$. 
This takes care of the unambiguity condition.
Since $G_i$ only depends on $u_i\oplus v_i$, it will be an XOR family and 
since it only evaluates to $1$ on a certificate, it will be nontrivial.

We now construct the certificate. Let $\Size(F) = m$, which means that there is a circuit that takes in 
$(x,y)$ as input and outputs $F(x,y)$ using at most $m$ constant fanin gates. 
The cell $u_i\oplus v_i$ will contain $c$ certificates, each certifying that the corresponding input to $F$ evaluates to correct bit of $i$.
For one instance of $F$, the certificate is constructed as follows. 
The certificate has to provide a full evaluation of the circuit of size $m$ on $(x,y)$ by providing
the correct values for the inputs and outputs of all $m$ gates.
The final gate should, of course, evaluate the the claimed output value for $F$. 
The inputs to the first level, which are inputs belonging to either Alice or Bob, should be
consistent with the true inputs that Alice and Bob hold. For a circuit of size $m$, a certificate of this sort has size $\tO(m)$ (with a log factor to account for describing the labels of gates), and hence the entire certificate has size $\tO(cm)$.

If the inputs are consistent with Alice's and Bob's input, and all the gates are evaluated correctly, then the output of the circuit will be $F(x,y)$ and the output string for all $c$ circuits will indeed be $F^c(\x,\y) = \ell$. If this output string is consistent with $i$, then $G_i$ accepts and otherwise rejects.

It is easy to see that $\mathcal{G}$ satisfies the first two properties we wanted.  It remains to upper bound $\Q(G_i)$.
As a warmup, note that the deterministic communication complexity of $G_i$ is at most $\tO(cm)$.
This is because Alice and Bob can simply send all of $u_i$ and $v_i$ to each other, which costs $\tO(cm)$ communication. They can then check that the their inputs are correct, the circuit evaluation is correct, and the circuits evaluate to $i$.

A similar algorithm, using Grover's algorithm to search for a discrepancy, yields the quantum algorithm.
Alice and Bob first check that the $O(cm)$ inputs in the circuits (there are $O(m)$ inputs per $F$, and there are $c$ copies of $F$) are consistent with their part of the input using $\tO(\sqrt{cm})$ communication using Grover's algorithm. They can then Grover search over all $cm$ gates to check if their inputs and outputs are consistent, which again takes  $\tO(\sqrt{cm})$ communication.
The final step is to check that the output bits equal $i$. This takes $\tO(\sqrt{c})$
communication using Grover search. Hence the total quantum communication complexity
of $G_i$ is $\tO(\sqrt{c m})=\tO(\sqrt{c\,\Size(F)})$.
\end{proof}

\lem{logrank_cheatsheet} and \lem{show_your_work} straightforwardly imply \thm{upper}.

\section{Lower bound on quantum communication complexity of lookup functions}
\label{sec:lower}

In this section, we prove our main theorem, which is the following:

\begin{theorem} \label{thm:qcs} Let $F: \X \times \Y \rightarrow \{0,1,*\}$ be a (partial) function, $c \ge 5\log(\Q_{1/3}(F))$ and $r \ge 1$ be an integer. Let $\G = \{G_0,\ldots,G_{2^c-1}\}$ be a nontrivial family of XOR functions where each $G_i:  \left( \X^c \times \{0,1\}^m \right) \times \left( \Y^c \times \{0,1\}^m \right) \rightarrow \{0,1\}$, and let $F_{\G}$ be the $(F,\G)$-lookup function. Let $\delta = \frac{1}{10^{9}cr^2}$. For any $1/3$-error $r$-round protocol $\Pi$ for $F_{\G}$, there exists a $\frac{1}{2} - \frac{\delta}{3}$-error protocol $\Pi'$ for $F$ such that 
$$
\QCC(\Pi') = O(\QCC(\Pi)) .
$$
\end{theorem}

Before proving this, we show how it implies the corollary used in \sec{separation}, which we restate.

\corlower*

\begin{proof}
Let $\Pi$ be a protocol for $F_\G$ with $\QCC(\Pi) = \Q(F)$. Then from  \thm{qcs}, we have $\Q_\eps(F) = O(\Q(F_\G))$, where $\eps=\frac{1}{2} - \frac{\delta}{3}$,  $\delta = \frac{1}{10^{9}cr^2}$, and $r\leq \QCC(\Pi) = \Q(F_\G))$ is the number of rounds in $\Pi$. Now from \thm{discQ}, we know that 
$\Q_\eps(F) =\Omega\bigl(\log \frac{1-2\eps}{\disc(F)}\bigr)$. Combining these with the  fact that $cr^2 = O(\Q(F_\G))$ we get
$$\Q(F_\G) =\Omega\(\log \frac{1-2\eps}{\disc(F)}\)= \Omega\(\log\(\frac{1}{\disc(F)}\) - \log(cr^2)\)=\Omega\(\log\(\frac{1}{\disc(F)}\) - \log \Q(F_\G)\),$$
which implies the statement to be proved.
\end{proof}

	\begin{figure}[t]
		\centering
		\begin{tikzpicture}[scale=1.025]
		\pgfmathsetmacro{\left}{0};
		\pgfmathsetmacro{\right}{13};
		\pgfmathsetmacro{\up}{8.4};
		\pgfmathsetmacro{\down}{0.0};
		\pgfmathsetmacro{\middle}{(\left+\right)/2)};
		\draw (\left,\down) -- (\left,\up+0.5) -- (\right,\up+0.5) -- (\right,\down) -- (\left,\down);
		\node (mu) at (\left+0.3,\up+0.2) {$\forall\mu$};
		\node[draw] (conclusion) at (\middle-4.2,\down-1.7) {$\exists\;\Pi'\textrm{ for }F\textrm{ with}\;\;\err(\Pi')\leq\frac{1}{2}-\frac{\delta}{3},	 \quad\QCC(\Pi')=O(\QCC(\Pi))$};
		\node[draw] (minimax) at (\middle, \down-0.6) {\fullbref{fact:equiv}};
		\node[draw,align=center] (ass1) at (\left-1.7,\up-3) {\eq{assumption1}\\ $\QCC(\Pi)\leq \delta 2^c\;?$};
		\node[draw,align=center] (ass2) at (\middle-4,\up-0.8) {\eq{assumption2}\\$\VR(XL,X\otimes W)\leq\frac{c\delta}{3}\; ?$};
		\node[draw,align=center] (ass3) at (\middle+3,\up-0.8) {\eq{assumption3} $\forall i \;\; \forall k \; \;$\\$\I(A_k U \tU X_{-i} Y_{-i}; L_i | X_i)_{\psi^k} \le \delta ,$ \\ $\I(B_k V \tV X_{-i} Y_{-i}; L_i | Y_i)_{\psi^k} \le \delta\; ?$};
		\node[draw] (X) at (\middle-1.5,\up-3) {\clm{low_info_CS}};
		\node[draw] (pi_mu) at (\middle,\down+0.8) {$\exists\;\Pi_\mu\;\;
			\err^\mu(\Pi_\mu)\leq\frac{1}{2}-\frac{\delta}{3},\;\;\QCC(\Pi_\mu)\leq\QCC(\Pi)+1$};
		\node[draw] (Z) at (\middle, \up-4.25) {\clm{fix_CS}};
		\node[draw] (Z2) at (\middle,\up-5.250) {\clm{quantum_cut_paste}};
		\node[draw] (contradiction) at (\middle, \up-6.5) {contradiction};
		\node[draw] (ass2no) at (\middle-4,\up-4.75) {\clm{biasmu}};
		\node[draw] (ass3no) at (\middle+3,\up-4.75) {\clm{lowinf}};
		\draw[-latex] (ass2) -- (X) node[midway,fill=white] {Yes};
		\draw[-latex] (ass3) -- (Z) node[midway,fill=white] {Yes};
		\draw[-latex] (ass2) -- (ass2no) node[near start,fill=white] {No};
		\draw[-latex] (ass3) -- (ass3no) node[near start,fill=white] {No};
		\draw[-latex] (ass2no.south) -- (ass2.south |- pi_mu.north);
		\draw[-latex] (ass3no.south) -- (ass3.south |- pi_mu.north);
		\draw[-latex] (ass1) -- (X) node[near start,fill=white] {Yes};
		\draw[-latex] (ass1.south) -- (ass1.south |- conclusion.north)
		node[near start,fill=white] {No};
		\draw[-latex] (X) -- (Z) node[near end,above] {\eq{X}};
		\draw[-latex] (Z) -- (Z2)  ;
		\draw[-latex] (Z2) -- (contradiction) node[midway,right] {\eq{Z}};
		\draw[-latex] (pi_mu.south) -- (pi_mu.south |- minimax.north);
		\draw[-latex] (pi_mu) -- (minimax);
		\draw[-latex] (minimax.south) -- (minimax.south |- conclusion.north);
		\end{tikzpicture}
		\caption{The structure of the proof of \thm{qcs}.
			Note that \clm{low_info_CS} and \clm{fix_CS}  only follow if both of their incoming
			arcs hold.}
		\label{fig:proof_figure}
	\end{figure}
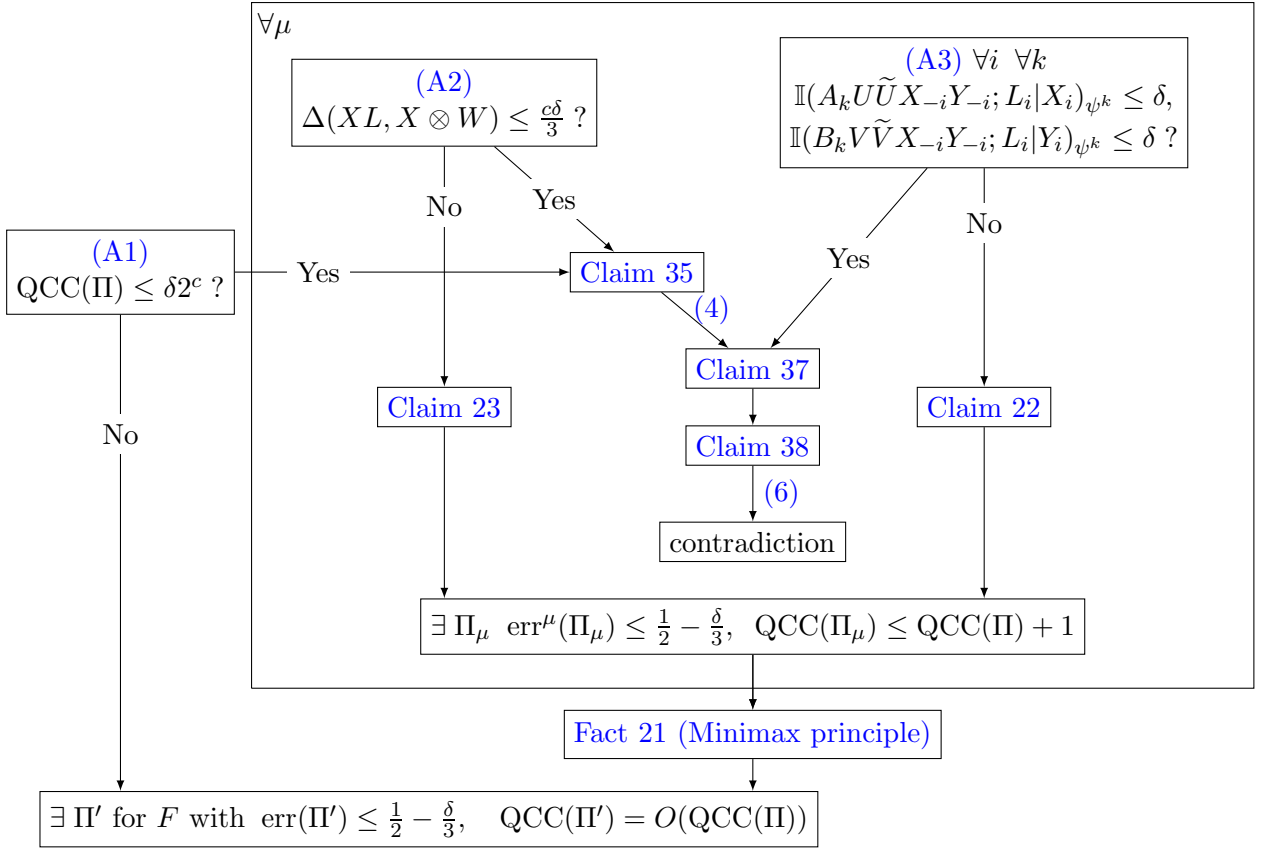

\begin{proof}[Proof of \protect{\thm{qcs}}]
We explain here the overall structure of the argument which is also displayed visually in \fig{proof_figure}.

\para{Rule out trivial protocols.} We first rule out the easy case where the protocol we are given, $\Pi$, has high quantum communication cost.  More precisely, we check if the following condition holds.
\begin{equation*}
\label{eq:assumption1}
\QCC(\Pi) < \delta 2^c .
\tag{A1}
\end{equation*}
If this does not hold then $\QCC(\Pi) \geq \delta 2^c = \Omega(\Q(F))$. By choosing the protocol whose communication complexity is $\Q(F)$, we obtain a protocol $\Pi'$ for $F$ with $ \QCC(\Pi') = \Q(F) = O(\QCC(\Pi))$ and we are done. Hence for the rest of the proof we may assume \eq{assumption1}.

\para{Protocols correct on a distribution.}
Instead of directly constructing a protocol $\Pi'$ for $F$ that is correct on all inputs with bounded error, we instead construct for every distribution $\mu$ on $\dom(F)$, a protocol $\Pi_\mu$ that does well on $\mu$ and then use \fct{equiv} to construct our final protocol. More precisely, for every $\mu$ over $\dom(F)$  we construct a protocol $\Pi_\mu$ for $F$ that has the following properties:
\begin{align}
\QCC(\Pi_\mu)  = \QCC(\Pi) + 1 \qquad \textrm{and} \qquad \err^\mu(\Pi_\mu) < 1/2-\delta/3. \label{eq:niceprot}
\end{align}

Hence for the remainder of the proof let $\mu$ be any distribution over $\dom(F)$ and our aim is to construct a protocol satisfying \eq{niceprot}.
\para{Construct a distribution for $\mathbf{F_{\G}}$.} Using the distribution $\mu$ on $\text{dom}(F)$, we now construct a distribution over the inputs to $F_{\G}$. Let the random variable $T$ be defined as follows:
$$
T := (X_1,\ldots,X_c, U_0,\ldots,U_{2^c-1}, Y_1,\ldots,Y_c, V_0,\ldots,V_{2^c-1}), 
$$
where for all $i \in [c]$, $X_i Y_i$ is distributed according to $\mu$ and independent of all other random variables and for $j \in \{0,\ldots,2^c-1\}$, $U_j V_j$ are uniformly distributed in $\{0,1\}^{2m}$ and independent of all other variables. 
For $i \in [c]$, we define $L_i := F(X_i, Y_i)$. We also define $X := (X_1,\ldots,X_c)$, $Y := (Y_1,\ldots,Y_c)$, $L := (L_1,\ldots,L_c)$, $U := (U_0,\ldots,U_{2^c-1})$ and $V := (V_1,\ldots,V_{2^c-1})$. Lastly, for $i \in [c]$, we define $X_{-i} := X_1,\ldots,X_{i-1},X_{i+1}, \ldots,X_c$ and $X_{<i} :=X_1,\ldots,X_{i-1}$. Similar definitions hold for $L$ and $Y$. 
Let $A_k, B_k$ be the registers of Alice and Bob after round $k$ of protocol $\Pi$. The total pure state after round $k$ can be written as follows:
$$
\ket{\psi_k}_{X \widetilde{X} U \widetilde{U} Y \widetilde{Y} V \widetilde{V} A_k B_k} = \sum_{x,u,y,v} \sqrt{\mu_T(x,u,y,v)} \ket{x x}_{X \widetilde{X}} \ket{u u}_{U \widetilde{U}} \ket{y y}_{Y \widetilde{Y}} \ket{v v}_{V \widetilde{V}} \ket{\psi^{x,u,y,v}_k}_{A_k B_k}
$$
Here $\mu_T$ is the distribution of the random variable $T$. $\widetilde{X}, \widetilde{U}, \widetilde{Y}, \widetilde{V}$ are registers that purify the classical inputs $X,U,Y,V$ respectively. 

\para{Rule out easy distributions $\mathbf{\mu}$.} We now show that if $\mu$ is such that the output of $F(X,Y)$ is predictable simply by looking at Alice's input $X$, then this distribution is easy and we can construct a protocol $\Pi_\mu$ that does well on this distribution since Alice can simply guess the value of $F(X,Y)$ after seeing $X$. More precisely, we check if the following condition holds.
\begin{equation*}
\label{eq:assumption2}
\VR(XL, X \otimes W)  \leq c \delta/3,
\tag{A2}
\end{equation*}
where $W$ is the uniform distribution on $\{0,1\}^c$. 

If the condition does not hold, we invoke \clm{biasmu} with $\epsilon =\delta/3$. Then we must be in case~(a) of this claim and hence we get the desired protocol $\Pi_\mu$.
Therefore we can assume \eq{assumption2} holds.

\para{Construct new protocols $\Pi_i$.} We now define a collection of protocols $\Pi_i$ for each $i \in [c]$. $\Pi_i$ is a protocol in which Alice and Bob receive inputs from $\text{dom}(F)$. We construct $\Pi_i$ as follows: Given the input pair $(X_i, Y_i)$ distributed according to $\mu$, Alice and Bob use shared entanglement $X_{-i} \widetilde{X_{-i}} Y_{-i} \widetilde{Y_{-i}}$ (Alice holds $X_{-i} \widetilde{X_{-i}}$ and Bob holds $Y_{-i} \widetilde{Y_{-i}}$), where $X_{-i} Y_{-i}$ are distributed according to $\mu^{\otimes c-1}$ and $\widetilde{X_{-i}}\widetilde{Y_{-i}}$ purify $X_{-i} Y_{-i}$ in a canonical way. They also use shared entanglement  $U \widetilde{U} V \widetilde{V}$ (Alice holds $U \widetilde{U}$ and Bob holds $ V \widetilde{V}$), where  $U$ and $V$ are uniformly distributed and $\widetilde{U}\widetilde{V}$ purify $UV$ in a canonical way. Note that Alice and Bob now have inputs $XU$ and $YV$ distributed according to $T$. They then run protocol $\Pi$. It is clear that for all $i \in [c]$, $\QCC(\Pi_i) = \QCC(\Pi)$.

\para{Rule out informative protocols $\Pi_i$.} If any of the protocols $\Pi_i$ that we constructed has a lot of information about $L_i$, then we can use~\clm{lowinf} to design a protocol for $F$. Hence, we can assume that for each $1 \le k \le r$, 
\begin{equation*}
\label{eq:assumption3}
\I(A_k U \tU X_{-i} Y_{-i}; L_i | X_i)_{\psi_k}, \I(B_k V \tV X_{-i} Y_{-i}; L_i | Y_i)_{\psi_k} \le \delta .
\tag{A3}
\end{equation*}

\para{Obtain a contradiction.} 
We have already established that \eq{assumption1}, \eq{assumption2}, and \eq{assumption3} must hold, otherwise we have obtained our protocol $\Pi_\mu$. We will now show that if  \eq{assumption1}, \eq{assumption2}, and \eq{assumption3}  simultaneously hold, then we obtain a contradiction. To show this, we use some claims that are proved after this theorem.

First we apply \clm{subadditivity} to get the following from \eq{assumption1} and \eq{assumption2}.
\begin{equation}
\forall k \in \{1,\ldots, r\}:  \quad \mathbb{E}_{x,l \leftarrow XL} \BR^2 \left(\psi^{x}_{k,B_k Y \widetilde{Y} V \widetilde{V} U_l}, \psi^{x}_{k,B_k Y \widetilde{Y} V \widetilde{V}} \otimes \psi_{ U_l} \right) \le \frac{q}{2^c} + \frac{c\delta}{3} . \label{eq:X}
\end{equation}
Here $q = \QCC(\Pi)/2$. Intuitively this claim asserts that for a typical $x$ and $\ell$, Bob (conditioned on $X=x$) has very little information about the cell $U_\ell$ at the end of round $k$, which is quantified by saying their joint state is close to being a product state. This would be false without assuming \eq{assumption1} because if there was no upper bound on the communication in $\Pi$, then Alice could simply communicate all of $U$, in which case Bob would have a lot of information about any $U_j$. We need \eq{assumption2} as well, since otherwise it is possible that the correct answer $\ell$ is easily predicted by Alice by looking at her input alone, in which case she can send over the contents of cell $U_\ell$ to Bob. A symmetric statement also follows with Alice and Bob interchanged. 

We then apply \clm{low_info_CS} to get the following from \eq{assumption3}.
\begin{equation}
\forall k \in \{1,\ldots, r\}: \quad \mathbb{E}_{x,l \leftarrow XL} \BR^2 \left(\psi^{x,l}_{k,B_k Y \widetilde{Y} V \widetilde{V} U_l}, \psi^{x,l}_{k,B_k Y \widetilde{Y} V \widetilde{V}} \otimes \psi_{U_l} \right) \le 3 \cdot \left( \frac{q}{2^c} + \frac{c\delta}{3} + 2 c \delta \right) . \label{eq:Y}
\end{equation}
Intuitively, this claim asserts that for a typical $x$ and $\ell$, Bob (conditioned on $X=x$ and $L=\ell$) has very little information about the cell $U_\ell$ at the end of round $k$, which is quantified by saying their joint state is close to being a product state. A symmetric statement also follows for Alice. 
Equation \ref{eq:Y} implies the following relation, which is proved in \clm{biased_output}: $\Pr_{x,y,l,u_l,v_l \leftarrow X,Y,L,U_L,V_L}[G_l(x,y,u_l,v_l) = \alpha(x,y)] \le 1/100$, where $\alpha(x,y)$ is either $0$ or $1$. We then proceed to apply Claim \ref{clm:fix_CS}.

We then apply \clm{quantum_cut_paste}, which uses \eq{X} and \eq{Y} and \clm{fix_CS}, to obtain the following. There exists, $x,y,l,\widetilde{u}_l, \widetilde{v}_l, \widetilde{\widetilde{u}}_l, \widetilde{\widetilde{v}}_l$ such that,
\begin{align} 
 \Delta \left((\psi_{r,A_r U_{-l} \tU_{-l}}^{x,y,l,\widetilde{u}_l, \widetilde{v}_l}, \psi_{r,A_r U_{-l} \tU_{-l}}^{x,y,l,\widetilde{u}_l,\widetilde{\widetilde{v}}_l}\right) &\le 1000r \cdot \sqrt{\left( \frac{q}{2^c} + \frac{c\delta}{3} + 2 c \delta \right)} < 0.1 , \nonumber \\
 G_l(x,y, \widetilde{u}_l, \widetilde{v}_l) = 1 &\text{ and } G_l(x,y, \widetilde{u}_l,  \widetilde{\widetilde{v}}_l) = 0. \label{eq:Z}
\end{align}
 We assume (w.l.o.g) that Alice gives the answer in round $r$. From above 
$$ |\Pr(\mbox{Alice outputs $1$ on } (x,y, \widetilde{u}_l, \widetilde{v}_l)) -  \Pr(\mbox{Alice outputs $1$ on } (x,y, \widetilde{u}_l, \widetilde{\widetilde{v}}_l)) | < 0.1 .$$
This is a contradiction since $G_l(x,y, \widetilde{u}_l, \widetilde{v}_l) = 1 \text{ and } G_l(x,y, \widetilde{u}_l,  \widetilde{\widetilde{v}}_l) = 0$ and the error of $\Pi$ on any input is at most $1/3$.

\para{Minimax argument.}
Note that in all branches where we did not reach a contradiction, we constructed a protocol satisfying \eq{niceprot}. Hence we constructed, for any $\mu$ over $\dom(F)$, a protocol $\Pi_\mu$ that satisfies \eq{niceprot}.  We now use~\fct{equiv} to complete the proof. 
\end{proof}

This completes the proof of the theorem, except for the claims \clm{subadditivity}, \clm{low_info_CS},  \clm{biased_output}, \clm{fix_CS}, and \clm{quantum_cut_paste} that we did not prove. We now prove these claims.

\subsection{Proof of claims}

\begin{claim}\label{clm:subadditivity}
Suppose $\QCC(\Pi) = 2q $ and $\Delta(XL, X \otimes W) \le \delta_1 $. Then 
$$
\mathbb{E}_{x,l \leftarrow XL} \BR^2 \left(\psi_{k,B_k Y \widetilde{Y} V \widetilde{V} U_l}^{x}, \psi^{x}_{k,B_k Y \widetilde{Y} V \widetilde{V}} \otimes \psi_{ U_l} \right) \le \frac{q}{2^c} + \delta_1.
$$
for all $1 \le k \le r$. Here $\psi_{U_l}$ is the maximally mixed state on the register $U_l$ (in other words a random variable which is uniformly distributed.)
\end{claim}

\begin{proof} We have
\begin{align*}
q &\ge \I(B_k Y \tY  V \tV ; U_0,\ldots,U_{2^c-1}|X)_{\psi_k} & ~\text{(\lem{qic<=2qcc})}\\
&\ge \sum_{l=0}^{2^c-1} \I(B_k Y \tY  V \tV ; U_l | X)_{\psi_k} & ~\text{(\fct{infoind})}\\
&=2^c \cdot \mathbb{E}_{x, l \leftarrow X \otimes W} \I(B_k Y \tY  V \tV ; U_l | X = x)_{\psi_k} \\
&\ge 2^c \cdot \mathbb{E}_{x, l \leftarrow X \otimes W} \BR^2 \left(\psi^{x}_{k,B_k Y \widetilde{Y} V \widetilde{V} U_l}, \psi^{x}_{k,B_k Y \widetilde{Y} V \widetilde{V}} \otimes \psi_{U_l} \right)  & (\text{\fct{IvsB}}) .
\end{align*}
This implies that 
$$
\mathbb{E}_{x,l \leftarrow X \otimes W} \BR^2 \left(\psi^{x}_{k,B_k Y \widetilde{Y} V \widetilde{V} U_l}, \psi^{x}_{k,B_k Y \widetilde{Y} V \widetilde{V}} \otimes \psi_{ U_l} \right) \le \frac{q}{2^c} .
$$
Since $\Delta(XL, X \otimes W) \le \delta_1$ and $\BR^2(\rho, \sigma) \le 1$ always, this proves the claim as well. 
\end{proof}

The next claim intuitively says that, if the communication cost of $\Pi$ is small, then at any point during the protocol, Bob's register has small information about the correct cheat sheet cell. 

\begin{claim}\label{clm:low_info_CS}
Assume in addition to the assumptions of \clm{subadditivity}, the following condition holds: for all $i \in [c]$, let 
$$
\I(A_k U \tU X_{-i} Y_{-i}; L_i | X_i)_{\psi_k} \le \delta .
$$
Then 
$$
\mathbb{E}_{x,l \leftarrow XL} \BR^2 \left(\psi_{k,B_k Y \widetilde{Y} V \widetilde{V} U_l}^{x,l}, \psi_{k,B_k Y \widetilde{Y} V \widetilde{V}}^{x,l} \otimes \psi_{U_l} \right) \le 3 \cdot \left( \frac{q}{2^c} + \delta_1 + 2 c \delta \right)
$$
for all $1 \le k \le r$. 
\end{claim}

\begin{proof}
We first prove that the register $A_k$ carries low information about $L$ i.e. 
$$
\I(A_k U \tU; L|X)_{\psi^k} \le c \delta .
$$ 
This follows from the following chain of inequalities:
\begin{align*}
\delta &\ge \I(A_k U \tU X_{-i} Y_{-i}; L_i | X_i)_{\psi_k} \\
&\ge  \I(A_k U \tU X_{-i} L_{<i}; L_i | X_i)_{\psi_k} & \text{(\fct{mono} and \fct{IvsB})}\\
&\ge  \I(A_k U \tU; L_i | L_{<i}, X)_{\psi_k} &\text{(\fct{barhopping})}.
\end{align*}
By summing the inequality over $i$, we get
\begin{align*}
c \delta &\ge \sum_{i=1}^c \I(A_k U \tU; L_i | L_{<i}, X)_{\psi_k} \\
&= \I(A_k U \tU; L|X)_{\psi_k} &  \text{(\fct{chain-rule})}.
\end{align*}
This implies using~\fct{IvsB}: 
\begin{align}
\mathbb{E}_{x,l \leftarrow XL} \BR^2 \left(\psi^{x,l}_{k,A_k U \tU}, \psi^{x}_{k,A_k U \tU} \right) \le c \delta . \label{eqn:ankit1}
\end{align}
Now consider the following two pure states (one conditioned on $x,l$ and the other conditioned on $x$):
$$
\ket{\psi^{x,l}}_{k,Y \tY V \tV U \tU A_k B_k} = \sum_{y,v,u} \sqrt{\mu_T(y,v,u|X=x,L=l)}  \ket{u u}_{U \widetilde{U}} \ket{y y}_{Y \widetilde{Y}} \ket{v v}_{V \widetilde{V}} \ket{\psi^{x,u,y,v}}_{k,A_k B_k}
$$
and 
$$
\ket{\psi^{x}}_{k,Y \tY V \tV U \tU A_k B_k} = \sum_{y,v,u} \sqrt{\mu_T(y,v,u|X=x)}  \ket{u u}_{U \widetilde{U}} \ket{y y}_{Y \widetilde{Y}} \ket{v v}_{V \widetilde{V}} \ket{\psi^{x,u,y,v}}_{k,A_k B_k} .
$$
The marginals of these states on the systems $A_k U \tU$ are close as shown above. Now by Uhlmann's theorem (\fct{uhlmann}), there exists a unitary acting on the systems $B_k Y \tY V \tV$ (and the unitary depends on $x,l$) $\U^{x,l}_{B_k Y \tY V \tV}$ s.t. 
$$
\BR^2 \left(\id_{A_kU \tU} \otimes \U^{x,l}_{B_k Y \tY V \tV} \ket{\psi^{x,l}}_{k,A_k U \tU B_k Y \tY V \tV}, \ket{\psi^{x}}_{k,A_k U \tU B_k Y \tY V \tV} \right) = \BR^2 \left(\psi^{x,l}_{k,A_k U \tU}, \psi^{x}_{k,A_k U \tU} \right) .
$$
The unitary $\U^{x,l}_{B_k Y \tY V \tV}$ should be intuitively thought of as implementing the operation of ``forgetting $L$". Hence Equation~\eqref{eqn:ankit1} gives us that: 
\begin{align}
\mathbb{E}_{x,l \leftarrow XL} \BR^2 \left(\id_{A_kU \tU} \otimes \U^{x,l}_{B_k Y \tY V \tV} \ket{\psi^{x,l}}_{k,A_k U \tU B_k Y \tY V \tV},  \ket{\psi^{x}}_{k,A_k U \tU B_k Y \tY V \tV}\right) \le c \delta .\label{eqn:ankit2}
\end{align}
For all $(x,\ell)$, define,
$$\phi^{x,\ell} = \id_{A_kU \tU} \otimes \U^{x,l}_{B_k Y \tY V \tV} \ket{\psi^{x,l}}_{k,A_k U \tU B_k Y \tY V \tV}.$$
Combining Equation~\eqref{eqn:ankit2} with the monotonicity of Bures metric (\fct{monotonicitydistance}), we obtain the following:
\begin{align}
\mathbb{E}_{x,l \leftarrow XL} \BR^2 \left( \phi^{x,l}_{k,B_k Y \tY V \tV U_l},  \psi^{x}_{k,B_k Y \tY V \tV U_l}\right) \le c \delta \label{eqn:ankit3}
\end{align}
and 
\begin{align}
\mathbb{E}_{x,l \leftarrow XL} \BR^2 \left(\phi^{x,l}_{k,B_k Y \tY V \tV},  \psi^{x}_{k,B_k Y \tY V \tV}\right) \le c \delta .\label{eqn:ankit4}
\end{align}
Furthermore, combining Equation~\eqref{eqn:ankit4} with Fact \ref{fact:prod}, we obtain:
\begin{align}
\mathbb{E}_{x,l \leftarrow XL}\BR^2 \left(\phi^{x,l}_{k,B_k Y \tY V \tV} \otimes \psi_{U_l},  \psi^{x}_{k,B_k Y \tY V \tV} \otimes \psi_{U_l}\right) \le c \delta . \label{eqn:ankit5}
\end{align}
\clm{subadditivity} gives us that:
\begin{align}
\mathbb{E}_{x,l \leftarrow XL} \BR^2 \left(\psi_{k,B_k Y \widetilde{Y} V \widetilde{V} U_l}^{x}, \psi^{x}_{k,B_k Y \widetilde{Y} V \widetilde{V}} \otimes \psi_{ U_l} \right) \le \frac{q}{2^c} + \delta_1  \label{eqn:ankit6} .
\end{align}
Now combining Equations~\eqref{eqn:ankit3}, \eqref{eqn:ankit5} and \eqref{eqn:ankit6} along with weak triangle inequality for square of Bures metric (\fct{triangle}) and \fct{monotonicitydistance}, we obtain:
\begin{align*}
\lefteqn{\mathbb{E}_{x,l \leftarrow XL} \BR^2 \left(\psi^{x,l}_{ k,B_k Y \widetilde{Y} V \widetilde{V} U_l}, \psi^{x,l}_{k,B_k Y \widetilde{Y} V \widetilde{V}} \otimes \psi_{U_l} \right) }\\
&=\mathbb{E}_{x,l \leftarrow XL} \BR^2 \left(\phi^{x,l}_{k,B_k Y \widetilde{Y} V \widetilde{V} U_l}, \phi^{x,l}_{k,B_k Y \widetilde{Y} V \widetilde{V}} \otimes \psi_{U_l} \right) \\
&\le 3 \cdot \left( \frac{q}{2^c} + \delta_1 + 2 c \delta \right) . \qedhere
\end{align*}
\end{proof}

\begin{claim}
\label{clm:biased_output}
Assuming the conclusion from \clm{low_info_CS}, it holds that 
$$
\Pr_{x,y,l,u_l,v_l \leftarrow X,Y,L,U_L,V_L}[G_l(x,y,u_l,v_l) = \alpha(x,y)] \le 1/100,
$$

where $\alpha(x,y)$ is either $0$ or $1$ for every $x,y$. 
\end{claim}

\begin{proof}
 Using monotonicity and partial measurement (\fct{mono} and \fct{relsubadd}), we have that: 
\begin{equation*}
\mathbb{E}_{x,y,l,u_l,v_l \leftarrow XYLU_LV_L} \BR^2 \left(\psi_{r, B_r }^{x,y,l,u_l,v_l}, \psi_{r, B_r}^{x,y,l,v_l}\right) \le 3 \cdot \left( \frac{q}{2^c} + \frac{c\delta}{3} + 2 c \delta \right) 
\end{equation*}

Let the output register be called $O$. Then, from our choice of parameters and monotonicity (\fct{mono}), above inequality implies 
\begin{equation}
\label{eqn:biaspsipsi}
\mathbb{E}_{x,y,l,u_l,v_l \leftarrow XYLU_LV_L} \BR^2 \left(\psi_{r, O }^{x,y,l,u_l,v_l}, \psi_{r, O}^{x,y,l,v_l}\right) \le 1/400
\end{equation}

Since protocol makes an error of at most $1/400$ (which can be assumed due to \fct{boost}), we have that 
\begin{equation}
\label{eqn:biaspsiG}
\mathbb{E}_{x,y,l,u_l,v_l \leftarrow XYLU_LV_L}\BR^2(\psi_{r, O}^{x,y,l,u_l,v_l}, \ketbra{G_l(x,y,u_l,v_l)})\leq 1/400.
\end{equation}

On the other hand, since the look-up function is an XOR family, we find that for a fixed $x,y$ (and hence a fixed $l$), 
\begin{align*}
\expec_{u_l\leftarrow U_L} \ketbra{G_l(x,y,u_l,v_l)} = &\Pr_{u_l,v_l\leftarrow U_l,V_l|x,y,l}[G_l(x,y,u_l,v_l) = 0]\ketbra{0} \\ +& \Pr_{u_l,v_l\leftarrow U_l,V_l|x,y,l}[G_l(x,y,u_l,v_l) = 1]\ketbra{1}.
\end{align*}
Define $p^0_{x,y,l} = \Pr_{u_l,v_l\leftarrow U_l,V_l|x,y,l}[G_l(x,y,u_l,v_l) = 0]$ and $p^1_{x,y,l} = \Pr_{u_l,v_l\leftarrow U_l,V_l|x,y,l}[G_l(x,y,u_l,v_l) = 1]$. Then above equation, along with Equation~\eqref{eqn:biaspsiG} implies that    
\begin{equation*}
\expec_{x,y,l,u_l,v_l \leftarrow XYLU_lV_l}\BR^2(\psi_{r, O}^{x,y,l,v_l}, p^0_{x,y,l}\ketbra{0} + p^1_{x,y,l}\ketbra{1})\leq 1/400
\end{equation*}
which in conjunction with Equation \ref{eqn:biaspsipsi} and triangle inequality gives us
\begin{equation}
\label{eqn:biasGG}
\expec_{x,y,l,u_l,v_l \leftarrow XYLU_lV_l}\BR^2(\ketbra{G_l(x,y,u_l,v_l)}, p^0_{x,y,l}\ketbra{0} + p^1_{x,y,l}\ketbra{1})\leq 1/100.
\end{equation}

This directly implies that we cannot have both $p^0_{x,y,l},p^1_{x,y,l}$ large. More formally, for every $x,y$, let $\alpha(x,y)$ be such that $p^{\alpha(x,y)}_{x,y,l} < p^{1-\alpha(x,y)}_{x,y,l}$. Then it is clear that \
$$\BR^2(\ketbra{G_l(x,y,u_l,v_l)}, p^0_{x,y,l}\ketbra{0} + p^1_{x,y,l}\ketbra{1}) > p^{\alpha(x,y)}_{x,y,l},$$
which in turn implies (when used in Equation \ref{eqn:biasGG}),
$$\expec_{x,y,l\leftarrow XYL}p^{\alpha(x,y)}_{x,y,l} = \expec_{x,y,l,u_l,v_l \leftarrow XYLU_lV_l}p^{\alpha(x,y)}_{x,y,l}\leq 1/100.$$ Recalling the definition of $p^{\alpha(x,y)}_{x,y,l}$, this immediately gives us
$$\expec_{x,y,l\leftarrow XYL}\Pr_{u_l,v_l\leftarrow U_l,V_l|x,y,l}[G_l(x,y,u_l,v_l) = \alpha(x,y)] \leq 1/100.$$
This completes the proof.
\end{proof}

\begin{claim}\label{clm:fix_CS} Assume that the assumptions of \clm{subadditivity} and \clm{low_info_CS} hold. In addition, 
$$
\I(B_k V \tV X_{-i} Y_{-i}; L_i | Y_i)_{\psi_k} \le \delta
$$
and 
$$
\Pr_{x,y,l,u_l,v_l \leftarrow X,Y,L,U_L,V_L}[G_l(x,y,u_l,v_l) = \alpha(x,y)] \le 1/100
$$
also hold for $\alpha(x,y)\in \{0,1\}$ for every $x,y$. Then there exist $x,y, l = l(x,y), \widetilde{u}_l,  \widetilde{v}_l,  \widetilde{\widetilde{u}}_l,  \widetilde{\widetilde{v}}_l$ s.t. the following conditions hold:
\begin{enumerate}
\item $G_l(x,y, \widetilde{u}_l, \widetilde{v}_l) = \alpha(x,y)$.
\item $G_l(x,y, \widetilde{u}_l,  \widetilde{\widetilde{v}}_l) = G_l(x,y, \widetilde{\widetilde{u}}_l,  \widetilde{\widetilde{v}}_l) = G_l(x,y, \widetilde{\widetilde{u}}_l,  \widetilde{v}_l) = 1-\alpha(x,y)$.
\item $\sum_{k=1}^r \BR \left(\psi^{x,y,u,v}_{k,B_k V_{-l} \tV_{-l}},  \psi^{x,y,v}_{k,B_k V_{-l} \tV_{-l}}\right) \le 80r \cdot \sqrt{\left(\frac{q}{2^c} + \delta_1 + 2 c \delta \right)}$, \\ 
for any choice of $(u,v) = (\widetilde{u}_l,  \widetilde{\widetilde{v}}_l), (\widetilde{\widetilde{u}}_l,  \widetilde{\widetilde{v}}_l), (\widetilde{\widetilde{u}}_l,  \widetilde{v}_l)$.
\item $\sum_{k=1}^r \BR \left(\psi^{x,y,u,v}_{k,A_k U_{-l} \tU_{-l}},  \psi^{x,y,u}_{k,A_k U_{-l} \tU_{-l}}\right) \le 80r \cdot \sqrt{\left( \frac{q}{2^c} + \delta_1 + 2 c \delta \right)}$, \\ 
for any choice of $(u,v) = (\widetilde{u}_l,  \widetilde{\widetilde{v}}_l), (\widetilde{\widetilde{u}}_l,  \widetilde{\widetilde{v}}_l), (\widetilde{\widetilde{u}}_l,  \widetilde{v}_l)$.
\end{enumerate}
\end{claim}

\begin{proof} By \clm{low_info_CS}, we have that for all $1 \le k \le r$, 
\begin{align*}
\mathbb{E}_{x,l \leftarrow XL} \BR^2 \left(\psi_{k, B_k Y \widetilde{Y} V \widetilde{V} U_l}^{x,l}, \psi_{k, B_k Y \widetilde{Y} V \widetilde{V}}^{x,l} \otimes \psi_{U_l} \right) \le 3 \cdot \left( \frac{q}{2^c} + \delta_1 + 2 c \delta \right) .
\end{align*}
By monotonicity of Bures metric (\fct{monotonicitydistance}), we get that 
\begin{align*}
\mathbb{E}_{x,l \leftarrow XL} \BR^2 \left(\psi_{k, B_k Y V_{-l} \widetilde{V}_{-l} U_l V_l}^{x,l}, \psi_{k, B_k Y V_{-l} \widetilde{V}_{-l} V_l}^{x,l} \otimes \psi^{U_l} \right) \le 3 \cdot \left( \frac{q}{2^c} + \delta_1 + 2 c \delta \right) .
\end{align*}
Note that in both the states above, the marginal state on registers $U_l V_l$ is maximally mixed. Then by the partial measurement property of the square of Bures metric, \fct{subadd}, we get that
\begin{align*}
\mathbb{E}_{x,y,l, u_l, v_l \leftarrow XYLU_L V_L} \BR^2 \left(\psi_{k, B_k V_{-l} \widetilde{V}_{-l}}^{x,y,l,u_l,v_l}, \psi_{k, B_k V_{-l} \widetilde{V}_{-l}}^{x,y,l,v_l}\right) \le 3 \cdot \left( \frac{q}{2^c} + \delta_1 + 2 c \delta \right) .
\end{align*}
Convexity of square gives us that
\begin{align}
\mathbb{E}_{x,y,l, u_l, v_l \leftarrow XYLU_L V_L} \BR \left(\psi_{k, B_k V_{-l} \widetilde{V}_{-l}}^{x,y,l,u_l,v_l}, \psi_{k, B_k V_{-l} \widetilde{V}_{-l}}^{x,y,l,v_l}\right) \le \sqrt{3} \cdot \sqrt{\left( \frac{q}{2^c} + \delta_1 + 2 c \delta \right)} . \label{eqn:ankit7}
\end{align}
Similarly we get that for all $1 \le k \le r$, 
\begin{align}
\mathbb{E}_{x,y,l, u_l, v_l \leftarrow XYLU_L V_L} \BR \left(\psi_{k, A_k U_{-l} \widetilde{U}_{-l}}^{x,y,l,u_l,v_l}, \psi_{k, A_k U_{-l} \widetilde{U}_{-l}}^{x,y,l,u_l}\right) \le \sqrt{3} \cdot \sqrt{\left( \frac{q}{2^c} + \delta_1 + 2 c \delta \right)} . \label{eqn:ankit8}
\end{align}
Summing Equations~\eqref{eqn:ankit7} and \eqref{eqn:ankit8} over $k$, we get the following:
\begin{align*}
\mathbb{E}_{x,y,l, u_l, v_l \leftarrow XYLU_L V_L} \sum_{k=1}^r \BR \left(\psi_{k, B_k V_{-l} \widetilde{V}_{-l}}^{x,y,l,u_l,v_l}, \psi_{k, B_k V_{-l} \widetilde{V}_{-l}}^{x,y,l,v_l}\right) \le 2r \cdot \sqrt{\left( \frac{q}{2^c} + \delta_1 + 2 c \delta \right)} .
\end{align*}
and
\begin{align*}
\mathbb{E}_{x,y,l, u_l, v_l \leftarrow XYLU_L V_L} \sum_{k=1}^r \BR \left(\psi_{k, A_k U_{-l} \widetilde{U}_{-l}}^{x,y,l,u_l,v_l}, \psi_{k, A_k U_{-l} \widetilde{U}_{-l}}^{x,y,l,u_l}\right) \le 2r \cdot \sqrt{\left( \frac{q}{2^c} + \delta_1 + 2 c \delta \right)} .
\end{align*}
Now by Markov's inequality, we can find $x,y,l = l(x,y)$ s.t. the following hold:
\begin{align}
&\Pr_{u_l,v_l \leftarrow U_l,V_l}[G_l(x,y,u_l,v_l) = \alpha(x,y)] \le 1/25 \label{eqn:ankit9} , \\
&\mathbb{E}_{ u_l, v_l \leftarrow U_l V_l} \sum_{k=1}^r \BR \left(\psi_{k, B_k V_{-l} \widetilde{V}_{-l}}^{x,y,l,u_l,v_l}, \psi_{k, B_k V_{-l} \widetilde{V}_{-l}}^{x,y,l,v_l}\right) \le 8r \cdot \sqrt{\left( \frac{q}{2^c} + \delta_1 + 2 c \delta \right)} , \label{eqn:ankit10} \\
&\mathbb{E}_{u_l, v_l \leftarrow U_l V_l} \sum_{k=1}^r \BR \left(\psi_{k, A_k U_{-l} \widetilde{U}_{-l}}^{x,y,l,u_l,v_l}, \psi_{k, A_k U_{-l} \widetilde{U}_{-l}}^{x,y,l,u_l}\right) \le 8r \cdot \sqrt{\left( \frac{q}{2^c} + \delta_1 + 2 c \delta \right)} \label{eqn:ankit11} .
\end{align}
Without loss of generality, assume that $\alpha(x,y) = 1$. Let us have the following two notations:
\begin{align*}
&\kappa_A(u_l,v_l) := \sum_{k=1}^r  \BR \left(\psi_{k, A_k U_{-l} \widetilde{U}_{-l}}^{x,y,l,u_l,v_l}, \psi_{k, A_k U_{-l} \widetilde{U}_{-l}}^{x,y,l,u_l}\right) , \\
&\kappa_\BR(u_l,v_l) := \sum_{k=1}^r \BR \left(\psi_{k, B_k V_{-l} \widetilde{V}_{-l}}^{x,y,l,u_l,v_l}, \psi_{k, B_k V_{-l} \widetilde{V}_{-l}}^{x,y,l,v_l}\right) .
\end{align*}
Recall that for $l = l(x,y)$, $G_l(x,y,u_l,v_l)$ is a non-trivial XOR function of the inputs $u_l,v_l$. So there exists a $t \in \{0,1\}^m$ s.t. $G_l(x,y,u,u \oplus t) = 1$ for all $u \in \{0,1\}^m$. Now we will choose $\widetilde{u}_l, \widetilde{\widetilde{u}}_l, \widetilde{\widetilde{v}}_l$ uniformly and independently from $\{0,1\}^m$ and set $\widetilde{v}_l = \widetilde{u}_l \oplus t$. Note that marginally, the distribution of $(u,v)$ is uniform over $\{0,1\}^m \times \{0,1\}^m$, for any choice of $(u,v) = (\widetilde{u}_l,  \widetilde{\widetilde{v}}_l), (\widetilde{\widetilde{u}}_l,  \widetilde{\widetilde{v}}_l), (\widetilde{\widetilde{u}}_l,  \widetilde{v}_l)$. Hence for any choice of $(u,v) = (\widetilde{u}_l,  \widetilde{\widetilde{v}}_l), (\widetilde{\widetilde{u}}_l,  \widetilde{\widetilde{v}}_l), (\widetilde{\widetilde{u}}_l,  \widetilde{v}_l)$, from Equations \eqref{eqn:ankit9}, \eqref{eqn:ankit10} and \eqref{eqn:ankit11}, we get the following:
\begin{align*}
&\text{Pr}_{\widetilde{u}_l, \widetilde{\widetilde{u}}_l, \widetilde{\widetilde{v}}_l}[G_l(x,y,u,v) = 1] \le 1/25 , \\
&\mathbb{E}_{ \widetilde{u}_l, \widetilde{\widetilde{u}}_l, \widetilde{\widetilde{v}}_l} \kappa_A(u,v) \le 8r \cdot \sqrt{\left( \frac{q}{2^c} + \delta_1 + 2 c \delta \right)} ,\\
&\mathbb{E}_{ \widetilde{u}_l, \widetilde{\widetilde{u}}_l, \widetilde{\widetilde{v}}_l} \kappa_\BR(u,v) \le 8r \cdot \sqrt{\left( \frac{q}{2^c} + \delta_1 + 2 c \delta \right)} .
\end{align*}
Now by a simple application of Markov's inequality, there exists a setting of $(\widetilde{u}_l, \widetilde{\widetilde{u}}_l, \widetilde{\widetilde{v}}_l)$ so that for any choice of $(u,v) = (\widetilde{u}_l,  \widetilde{\widetilde{v}}_l), (\widetilde{\widetilde{u}}_l,  \widetilde{\widetilde{v}}_l), (\widetilde{\widetilde{u}}_l,  \widetilde{v}_l)$, 
\begin{align*}
&G_l(x,y,u,v) = 0 , \\
& \kappa_A(u,v) \le 80r \cdot \sqrt{\left( \frac{q}{2^c} + \delta_1 + 2 c \delta \right)} , \\
& \kappa_\BR(u,v) \le 80r \cdot \sqrt{\left( \frac{q}{2^c} + \delta_1 + 2 c \delta \right)} .
\end{align*}
This completes the proof. Note that we chose $\widetilde{v}_l$ so that $G_l(x,y,\widetilde{u}_l, \widetilde{v}_l) = 1$. 
\end{proof}

The next claim will follow from the quantum-cut-and-paste lemma applied to \clm{low_info_CS}.

\begin{claim}\label{clm:quantum_cut_paste} Assume that the assumptions of \clm{subadditivity}, \clm{low_info_CS} and \clm{fix_CS} hold. Then for the $x,y,l,\widetilde{u}_l, \widetilde{v}_l, \widetilde{\widetilde{u}}_l, \widetilde{\widetilde{v}}_l$ in \clm{fix_CS}, it holds that
$$
 \Delta \left((\psi_{r,A_r U_{-l} \tU_{-l}}^{x,y,l,\widetilde{u}_l, \widetilde{v}_l}, \psi_{r,A_r U_{-l} \tU_{-l}}^{x,y,l,\widetilde{u}_l,\widetilde{\widetilde{v}}_l}\right) \le 1000r \cdot \sqrt{\left( \frac{q}{2^c} + \delta_1 + 2 c \delta \right)} .
$$
\end{claim}

\begin{proof}
Let us define the following registers: $\widetilde{A}_k := A_k U_{-l} \tU_{-l}$ and $\widetilde{B}_k := B_k V_{-l} \tV_{-l}$. Also we will define the following:
\begin{align*}
&\delta_{k,A} := \BR \left( \psi_{k,\widetilde{A}_k}^{x,y,\widetilde{\widetilde{u}}_l, \widetilde{v}_l},  \psi_{k,\widetilde{A}_k}^{x,y,\widetilde{\widetilde{u}}_l, \widetilde{\widetilde{v}}_l}\right) ,\\
&\delta_{k,B} := \BR \left( \psi_{k,\widetilde{B}_k}^{x,y,\widetilde{u}_l, \widetilde{\widetilde{v}}_l},  \psi_{k,\widetilde{B}_k}^{x,y,\widetilde{\widetilde{u}}_l, \widetilde{\widetilde{v}}_l}\right) .
\end{align*}
By the triangle inequality for Bures metric \fct{triangle}, 
\begin{align}
&\delta_{k,A} \le \BR \left( \psi_{k,\widetilde{A}_k}^{x,y,\widetilde{\widetilde{u}}_l, \widetilde{v}_l},  \psi_{k,\widetilde{A}_k}^{x,y,\widetilde{\widetilde{u}}_l}\right) + \BR \left( \psi_{k,\widetilde{A}_k}^{x,y,\widetilde{\widetilde{u}}_l,  \widetilde{\widetilde{v}}_l},  \psi_{k,\widetilde{A}_k}^{x,y,\widetilde{\widetilde{u}}_l}\right) \label{eqn:ankit12} ,\\
&\delta_{k,B} \le \BR \left( \psi_{k,\widetilde{B}_k}^{x,y,\widetilde{u}_l, \widetilde{\widetilde{v}}_l},  \psi_{k,\widetilde{B}_k}^{x,y,\widetilde{\widetilde{v}}_l}\right) + \BR \left( \psi_{k,\widetilde{A}_k}^{x,y,\widetilde{\widetilde{u}}_l,  \widetilde{\widetilde{v}}_l},  \psi_{k,\widetilde{A}_k}^{x,y,\widetilde{\widetilde{v}}_l}\right) \label{eqn:ankit13} .
\end{align}
Combining Equations~\eqref{eqn:ankit12}, \eqref{eqn:ankit13} and \clm{fix_CS}, we get the following:
\begin{align*}
&\sum_{k=1}^r \delta_{k,A} \le 160r \cdot \sqrt{\left( \frac{q}{2^c} + \delta_1 + 2 c \delta \right)} ,\\
&\sum_{k=1}^r \delta_{k,B} \le 160r \cdot \sqrt{\left( \frac{q}{2^c} + \delta_1 + 2 c \delta \right)} .
\end{align*}
Note that the state $\psi_{k,\widetilde{A}_k, \widetilde{B}_k}^{x,y,u,v}$ is a pure state for every $k,x,y,u,v$. Also for a fixed $x,y$, these states can be formed by a quantum protocol $\Pi'$ where Alice gets the input $u$ and Bob gets the input $v$ (since they are originally formed by running the protocol $\Pi$ and $U_{-l} \tU_{-l}$ and $V_{-l} \tV_{-l}$ are registers that can be owned by Alice and Bob respectively at the start of $\Pi'$). Hence we can apply \lem{quantum_cut_paste} (by setting $u = \widetilde{\widetilde{u}}_l$, $u' = \widetilde{u}_l$, $v = \widetilde{\widetilde{v}}_l$, $v' = \widetilde{v}_l$) to conclude that 
\begin{align*}
\BR \left(\psi_{r,\widetilde{A}_r}^{x,y,l,\widetilde{u}_l, \widetilde{v}_l}, \psi_{r,\widetilde{A}_r}^{x,y,l,\widetilde{u}_l,\widetilde{\widetilde{v}}_l} \right) 
&\le 2 \sum_{k=1}^r \left( \delta_{k,A} + \delta_{k,B}\right) \\
&\le 640r \cdot \sqrt{\left( \frac{q}{2^c} + \delta_1 + 2 c \delta \right)} .
\end{align*} 
Now the proof is finished by \fct{deltabures} and monotonicity of trace distance (\fct{monotonicitydistance}).
\end{proof}
 
\section{Conclusion and open problems}
\label{sec:conclusion}

We prove a nearly quadratic separation between the log of approximate rank and quantum communication complexity for a family of total functions, which is also the first superlinear separation between these two measures. Our separation is based on a lookup function constructed from the inner product function. To prove the lower bound on the quantum communication complexity of this lookup function, we prove a general purpose cheat sheet theorem for quantum communication complexity. We also prove a general theorem about an upper bound on log of approximate rank of lookup functions based on the circuit size of the base function. This proves the upper bound for an appropriate lookup function on inner product because the inner product function has a linear size circuit.

Several interesting open problems arise out of our work. We state some of them here:

\begin{enumerate}
\item Can we eliminate the round dependence in \thm{qcs}? Can we prove a similar result for quantum information complexity instead of quantum communication complexity, thereby separating quantum information complexity from log of approximate rank? 
\item Can we separate the quantum partition bound \cite{Laplante2012} from quantum communication complexity? Is the quantum partition bound a stronger lower bound measure than log of approximate rank?
\item Can we prove some sort of cheat sheet theorem for log of approximate rank? A simpler question might be to prove that for the inner product function on $n$ bits, any lookup function contructed using a nontrivial XOR family of functions has log of approximate rank at least $\Omega(\sqrt{n})$. 
\end{enumerate}

\section*{Acknowledgements}
We thank Aleksandrs Belovs, Mika G\"o\"os, and Miklos Santha for interesting discussions during the writing of \cite{ABB+16a}.
This work is partially supported by ARO grant number W911NF-12-1-0486, by the Singapore Ministry of Education and the National Research Foundation, also through NRF RF Award No. NRF-NRFF2013-13, and the Tier 3 Grant ``Random numbers from quantum processes'' MOE2012-T3-1-009.
This preprint is MIT-CTP \#4857.


\DeclareUrlCommand{\Doi}{\urlstyle{sf}}
\renewcommand{\path}[1]{\small\Doi{#1}}
\renewcommand{\url}[1]{\href{#1}{\small\Doi{#1}}}
\newcommand{\eprint}[1]{\href{http://arxiv.org/abs/#1}{\small\Doi{#1}}}
\bibliographystyle{alphaurl}
\phantomsection\addcontentsline{toc}{section}{References} 
\bibliography{cheat}

\end{document}